\documentclass[aps,superscriptaddress,twocolumn,showpacs,prl]{revtex4-2}

\usepackage{amsmath}
\usepackage{latexsym}
\usepackage{amssymb}
\usepackage{graphicx}
\usepackage[colorlinks=true, citecolor=blue, urlcolor=blue]{hyperref}
\usepackage{float}
\usepackage{amsfonts}
\usepackage{textcomp}
\usepackage{mathpazo}
\usepackage{blindtext,changes}

\usepackage{bbm}

\usepackage{xcolor}
\definecolor{myurlcolor}{rgb}{0,0,0.4}
\definecolor{mycitecolor}{rgb}{0,0.5,0}
\definecolor{myrefcolor}{rgb}{0.5,0,0}
\usepackage{hyperref}
\hypersetup{colorlinks,
linkcolor=myrefcolor,
citecolor=mycitecolor,
urlcolor=myurlcolor}

\usepackage[draft]{fixme}
\usepackage{amsmath,bbm}%,mathrsfs}
\usepackage{scrextend}
\usepackage{graphicx}
\usepackage{amsfonts}
\usepackage{amssymb}
\usepackage{amsmath,amssymb,amsthm,verbatim,graphicx,bbm}
\usepackage{mathrsfs}
\usepackage{mathtools}
\usepackage{color,xcolor,longtable}
\def\be{\begin{equation}}
\def\ee{\end{equation}}
\def\ben{\begin{eqnarray}}
\def\een{\end{eqnarray}}
\def\eea{\end{array}}
\def\bea{\begin{array}}

\newcommand{\Tr}[1]{\mathrm{Tr}#1}
\newcommand{\bei}{\begin{itemize}}
\newcommand{\eei}{\end{itemize}}
\newcommand{\ket}[1]{|#1\rangle}
\newcommand{\bra}[1]{\langle#1|}

\newcommand{\proj}[1]{\ket{#1}\!\bra{#1}}

\newcommand{\I}{\mathbbm{1}}

\newcommand{\dl}[1]{\left|\!\left|#1\right|\!\right|}

\renewcommand{\emph}[1]{\textbf{#1}}

\makeatletter
\newtheorem*{rep@theorem}{\rep@title}
\newcommand{\newreptheorem}[2]{%
\newenvironment{rep#1}[1]{%
 \def\rep@title{#2 \ref{##1}}%
 \begin{rep@theorem}}%
 {\end{rep@theorem}}}
\makeatother

\theoremstyle{plain}
\newtheorem{thm}{Theorem}%[section]
\newtheorem*{thm*}{Theorem}
\newreptheorem{thm}{Theorem}

\newtheorem{fakt}{Fact}

\newtheorem{cor}[thm]{Corollary}

\theoremstyle{definition}

\theoremstyle{remark}

\usepackage[T1]{fontenc}

\begin{document}

%%%%%%%%%%%%%%%%%%%%%%%%%%%%%%%%%%%%%%%%%%%%%%%%%%%%%%%%%%%%%%%%%%%

\title{A universal scheme to self-test any quantum state or measurement}
\author{Shubhayan Sarkar}
\email{shubhayan.sarkar@ug.edu.pl}
\affiliation{Laboratoire d’Information Quantique, Université libre de Bruxelles (ULB), Av. F. D. Roosevelt 50, 1050 Bruxelles, Belgium}
\affiliation{Institute of Informatics, Faculty of Mathematics, Physics and Informatics,
University of Gdansk, Wita Stwosza 57, 80-308 Gdansk, Poland}

\author{Alexandre C. Orthey, Jr.}
\affiliation{Institute of Fundamental Technological Research, Polish Academy of Sciences, Pawi\'nskiego 5B, 02-106 Warsaw, Poland.}
\affiliation{Faculty of Mathematics, Informatics and Mechanics, University of Warsaw, ul. Banacha 2, 02-097 Warsaw, Poland.}
\affiliation{Center for Theoretical Physics, Polish Academy of Sciences, Aleja Lotnik\'{o}w 32/46, 02-668 Warsaw, Poland}

\author{Remigiusz Augusiak}
\affiliation{Center for Theoretical Physics, Polish Academy of Sciences, Aleja Lotnik\'{o}w 32/46, 02-668 Warsaw, Poland}

%%%%%%%%%%%%%%%%%%%%%%%%%%%%%%%%%%%%%%%%%%%%%%%%%%%%%%%%%%%%%%%%%%%
\begin{abstract}	
%%%%%%%%%%%%%%%%%%%% version 1 of abstract %%%%%%%%%%%%

The emergence of quantum devices has raised a significant issue: how to certify the quantum properties of a device without placing trust in it. To characterise quantum states and measurements in a device-independent way, up to some degree of freedom, we can make use of a technique known as self-testing. While schemes have been proposed to self-test all pure multipartite entangled states (up to complex conjugation) and real local projective measurements, little has been done to certify mixed entangled states, composite or non-projective measurements. By employing the framework of quantum networks, we propose a scheme for self-testing (up to complex conjugation) arbitrary extremal 
measurements, including the projective ones, but also, in an indirect way, any quantum state, including the mixed ones and any quantum measurement, including non-extremal ones. The quantum network considered in this work is the simple star network, which is implementable using current technologies. For our purposes, we also construct a scheme that can be used to self-test the two-dimensional tomographically complete set of measurements with an arbitrary number of parties.
\end{abstract}

%%%%%%%%%%%%%%%%%%%%%%%%%%%%%%%%%%%%%%%%%%%%%%%%%%%%%%%%%%%%%%%%%%%

\maketitle

%%%%%%%%%%%%%%%%%%%%%%%%%%%%%%%%%%%%%%%%%%%%%%%%%%%%%%%%%%%%%%%%%%%
\textit{Introduction.}---Demonstrating the superiority of quantum systems over classical ones in tasks like computation and cryptography has been one of the most important advancements in quantum information theory. Quantum protocols often offer greater efficiency and security than their classical counterparts. However, to implement these protocols, it is necessary to certify the \textit{quantumness} of devices involved. The traditional methods of certifying a quantum state or a measurement rely on trusting the devices used, therefore referred to as device-dependent schemes. For instance, any quantum state can be device-dependently verified through, e.g., the quantum tomography if one has complete knowledge of the measurement apparatuses. However, from a practical perspective, it is unreasonable to completely trust the device. Therefore, it is desirable to devise certification schemes that require less or even minimal assumptions about the devices used such as for instance the device-independent ones.
%
%\replaced{that require minimal assumptions about the devices, known as device-independent %(DI) certification schemes.}{where the assumptions made on the devices are minimal, also referred to as device-independent (DI) certification schemes}. 
%
These schemes enable the verification of non-trivial quantum properties of an unknown device solely through the statistical data it produces. The essential resource underlying any DI scheme is Bell nonlocality---the existence of quantum correlations that cannot be reproduced by any local (or simply classical) hidden-variable model \cite{Bell,Bell66,NonlocalityReview}. In particular, observing a violation of a Bell inequality certifies, in a fully DI manner, that the device operates on an entangled state.

Interestingly, the observation of Bell nonlocality in the data produced by a given device can provide a lot more information about its internal working than just the presence of entanglement. The strongest and most complete form of such DI certification is self-testing \cite{Mayers_selftesting, Yao}. It enables almost full up to certain well understood equivalences characterization of the underlying quantum state and measurements performed on it. 
%
%certification of DI certification is known as self-testing. First introduced %in \cite{Mayers_selftesting, Yao}, self-testing is a DI way of certifying %quantum states and measurements by requiring minimal assumptions about the %underlying device except for its adherence to quantum theory. 
%
In recent years, a plethora of schemes have been proposed to self-test pure quantum states or quantum measurements (see, e.g., Refs. \cite{Scarani,Reichardt_nature,Mckague_2014,Wu_2014,Bamps,All,chainedBell, Projection,Jed1,prakash,Armin1,sarkar,sarkaro2,Allst,sarkarPRL}). In particular, Refs. \cite{Projection} and \cite{Allst} propose self-testing strategies for any pure entangled bipartite or multipartite states, respectively, where the second method is based on the quantum networks scenario. The only scheme to self-test a mixed entangled state, in particular, a bound entangled state, was only recently proposed in Ref. \cite{sarkarPRL} (see nevertheless Refs. \cite{subspaces1,subspaces2}). As for quantum measurements, apart from a bunch of schemes allowing to certify various examples of quantum measurements \cite{Armin1, Jed1, random1, sarkar, chen1}, including the composite ones \cite{Marco, JW2, NLWEsupic, sarkarPRL}, a general method to certify any real local projective measurement has only very recently been designed in Ref. \cite{chen1}. Despite this progress, there exists no single unified scheme that allows for certification of quantum measurements and quantum states (pure or mixed). Moreover, no general scheme for certification of composite measurements based on quantum networks has been proposed.

%to self-test most of them. Further on, there is no self-testing scheme for non-%projective measurements acting on non-qubit Hilbert spaces (see nevertheless %Ref. \cite{sarkar3} for a one-sided DI scheme). 

The main goal of this work is to tackle these challenges by providing a universal scheme that can self-test (up to complex conjugation) any extremal generalized quantum measurement (POVM) on any finite-dimensional Hilbert space. Since all projective measurements are extremal, they are included as well. We then extend the scheme to indirectly self-test any non-extremal quantum measurement. On the other hand, our method also enables DI-certified remote preparation of quantum states, which can also be seen as an indirect way of device-independently certifying of any quantum state, even mixed ones. Overall, our approach offers a unified way to self-test both arbitrary quantum states and measurements.

We utilise the framework of quantum networks consisting of an arbitrary number of parties and sources (as depicted in Fig. \ref{fig1} and explained later in the text). We first propose a Bell inequality that can be used to self-test the two-dimensional tomographically complete set of measurements with arbitrary number of parties. % in \textbf{a single-shot way}. 
Based on this, we are then able to certify any extremal POVM and eventually any quantum state. It is worth mentioning that, unlike the standard self-testing statements, our scheme requires additional causality constraints that are natural in quantum networks, in particular, that the sources are statistically independent. Our scheme is simple and can be implemented using the current infrastructure of quantum networks \cite{netexp1,netexp2,netexp3}. 

Before proceeding toward the main results, let us first introduce the scenario and relevant notions required throughout this work.

\textit{Preliminaries.---}The quantum network scenario under scrutiny is composed of an external group of $N$ Alices, denoted $A_i$ $(i=1,\ldots, N)$, and a central party $E$ called Eve. All the $N+1$ parties are spatially separated. Correspondingly, $N$ independent sources distribute bipartite quantum states---generally entangled---among the parties (cf. Fig. \ref{fig1}). These states are denoted by $\rho_{A_iE_i}$, where the subsystems $A_i$ will be measured by the external parties, whilst the other subsystems $E_i$ go to Eve. Within quantum theory, the joint state of statistically independent sources is represented as a product of the individual states, that is,
\begin{eqnarray}
    \rho_{AE}=\bigotimes_{i=1}^{N}\rho_{A_iE_i},
\end{eqnarray}
where we denoted $A\coloneqq A_1\ldots A_N$ and $E\coloneqq E_1\ldots E_N$. On their shares of the joint state $\rho_{AE}$, each external party $A_i$ performs one of three available measurements, each with two outcomes; the measurement choices and outcomes of party $A_i$ are denoted $x_i=0,1,2$ and $a_i=0,1$, respectively. At the same time, the central party Eve freely chooses to perform one of two measurements, where the first one has $2^N$ outcomes, whereas the second one has $K\leq 2^N$ outcomes. 
We denote Eve's measurement choices by $e=0,1$ and outcomes by $l=0,\ldots,2^N-1$ for $e=0$ and $l=0,\ldots,K-1$ for $e=1$. Let us add here that the first Eve's measurement is used to 
certify the external parties' measurements as well as the states prepared by 
the sources whereas the second measurement is the one to be certified.
%
%later on, for convenience, we also represent Eve's outcomes $l$  with the %binary representation, i.e., as a sequence of $N$ numbers $l_1\ldots l_N:=l$ %taking values $l_i\in\{0,1\}$. 
%
Importantly, the parties cannot communicate classically during the experiment.

After running many rounds of the above measurements, all the parties can bring their data together to reconstruct the set of probability distributions $\vec{p}=\{p(\mathbf{a}l|\mathbf{x}e)\}$ which describe the correlations that are produced between all the observers, where each $p(\mathbf{a}l|\mathbf{x}e)$ is the probability of obtaining outcomes $a_1\ldots a_N=:\mathbf{a}$ by all the Alices and $l$ by Eve after they perform the measurements labelled by $x_1\ldots x_N=:\mathbf{x}$ and $e$, respectively. From Born's rule, we can write
\begin{equation}\label{probs}
p(\mathbf{a}l|\mathbf{x}e)=\Tr\left[\rho_{AE}\left(\bigotimes_{i=1}^{N} M_{a_i|x_i} \otimes R_{l|e}\right)\right],
\end{equation}
where $M_{i,x_i}=\{M_{a_i|x_i}\}$, for every $i$, and $E_{e}=\{R_{l|e}\}$ are the measurements performed by parties $A_i$ and $E$. The measurement elements $M_{a_i|x_i}$ and $R_{l|e}$ are positive semi-definite and sum up to the identity on the respective Hilbert space for every measurement choice $x_i$ or $e$ of all $N+1$ parties. 

\begin{figure}[t]
    \centering
    \includegraphics[width=0.75\linewidth]{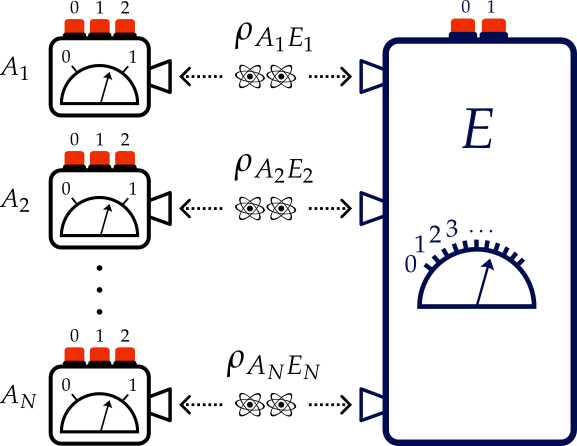}
    \caption{\textbf{Depiction of the quantum network scenario.} 
    It consists of $N+1$ parties, namely, $A_i$ $(i=1,\ldots,N)$, and $E$, and $N$ independent sources distributing bipartite quantum states $\rho_{A_iE_i}$ among the parties as shown in the figure. The central party $E$ shares quantum states with each one of the other external parties $A_i$. While each $A_i$ has three inputs and two outcomes, $E$ has two inputs, The first Eve's measurement has $2^N-1$ outcomes, whereas the second one has $K$ outcomes.}
    \label{fig1}
\end{figure}

To better understand the observed correlations, expressing them in terms of the expected values of external observables can be helpful. These values are defined as
\begin{equation}\label{obspic}
    \langle A_{1,x_1}\ldots A_{N,x_N} R_{l|e} \rangle
     = \sum_{a_1\ldots a_N=0}^{1} (-1)^{\sum_{i=1}^N a_i} p(\mathbf{a}l|\mathbf{x} e).
\end{equation}
It is important to note that by using Eq. \eqref{probs}, these expectation values can be expressed as $\langle A_{1,x_1}\ldots A_{N,x_N} R_{l|e} \rangle = \Tr[(\bigotimes_{i=1}^{N} A_{i,x_i}) \otimes R_{l|e}\rho_{AE}]$, where $A_{i,x_i}$ are quantum operators defined via the measurement elements as $A_{i,x_i}=M_{0|x_i}-M_{1|x_i}$ for every $x_i$ and $i$. Notice that in a particular case when the measurement $\{M_{0|x_i},M_{1|x_i}\}$ is projective, the corresponding operator $A_{x_i}$ becomes a standard hermitian observable with $\pm1$ eigenvalues. For our convenience, in the above representation Eve's measurements remain represented in terms of the measurement elements $R_{l|e}$ instead of observables.

\textit{Self-testing.---}Let us introduce the idea of self-testing by referring back to the scenario shown in Fig. \ref{fig1}. First, we assume that the measurements conducted by the parties and the states $\rho_{A_iE_i}\in\mathcal{L}(\mathcal{H}_{A_i}\otimes\mathcal{H}_{E_i})$ for $i=1,\ldots,N$ are unknown. The only knowledge that Alices and Eve
have about the whole system is encoded in the observed correlations $\vec{p}$.
It is worth pointing out already now that since the dimensions of the local Hilbert spaces $\mathcal{H}_{A_i}$ and $\mathcal{H}_{E_i}$ are unspecified, we can utilise the standard dilation argument and presume that the shared states are pure, i.e., $\rho_{A_iE_i}=\proj{\psi_{A_iE_i}}$. On the same ground we can also assume that the measurements of the external parties are projective. However, this assumption can be dropped by employing a more general sum of squares decomposition of the Bell operators corresponding to the inequalities (\ref{BE1Nm}) which take into account that the measurements are not projective.

Let us then consider another, reference experiment giving rise to the same correlations $\vec{p}$ which is performed on known quantum states shared by the parties 
$\ket{\psi'_{A_iE_i}}\in\mathcal{H}_{A_i'}\otimes\mathcal{H}_{E_i'}$ and with known  measurements $A_{i,x_i}'$ and $E'_e$ of Alices and Eve, respectively, where $\mathcal{H}_{A_i'}$ and $\mathcal{H}_{E_i'}$ are some Hilbert spaces of known dimension. The task of self-testing is to deduce from the observed $\vec{p}$ that the actual experiment is equivalent to the reference one in the following sense: (i) the local Hilbert spaces admit the product form $\mathcal{H}_{A_i}=\mathcal{H}_{A_i'}\otimes\mathcal{H}_{A_i''}$ and $\mathcal{H}_{E_i}=\mathcal{H}_{E_i'}\otimes\mathcal{H}_{E_i''}$ for some auxiliary Hilbert spaces $\mathcal{H}_{A_i''}$ and $\mathcal{H}_{E_i''}$; and (ii) there are local unitary operations  $U_{A_i}:\mathcal{H}_{A_i}\to \mathcal{H}_{A_i'}\otimes\mathcal{H}_{A_i''}$ and $U_{E_i}:\mathcal{H}_{E_i}\to \mathcal{H}_{E_i'}\otimes\mathcal{H}_{E_i''}$
such that
\begin{equation}\label{ststate}
 (U_{A_i}\otimes U_{E_i})\ket{\psi_{A_iE_i}}=\ket{\psi_{A_i'E_i'}'}\otimes\ket{\xi_{A_i''E_i''}},
 %
 %(U_{A_i}\otimes %U_{E_i})\ket{\psi_{A_iE_i}}=\ket{\psi'_{A_i'E_i'}}\otimes\ket{\xi_{A_i'%'E_i''}},
\end{equation}
where  $\ket{\xi_{A_i''E_i''}}\in\mathcal{H}_{A_i''}\otimes\mathcal{H}_{E_i''}$ is some auxiliary quantum state, and
\begin{equation}\label{stmea}
U_{A_i}\,A_{i,x_i}\,U_{A_i}^{\dagger}=A_{i,x_i}'\otimes\mathbbm{1}_{A_i''},\quad U_{E}\, R_{l|e}\,U_{E}^{\dagger}=R'_{l|e}\otimes\I_{E''},
\end{equation}
where $\mathbbm{1}_{E''}$ is the identity acting on the Eve's auxiliary system $\mathcal{H}_{E''}$ and $U_{E}=\bigotimes_{i=1}^N U_{E_i}$ and we use the notation $E''=E_1''\ldots E_N''$. 

If the above conditions (i) and (ii) are met one says that 
the reference state and measurements are self-tested in the actual experiment from the observed correlations. What is more, for any outcome $l$ of Eve's measurement, the post-measurement states $\rho^{l|e}_{A}$ held by the external parties $A_i$, 
%
%\begin{eqnarray}
 %\tilde{\rho}^{l|e}_{A}=\frac{1}{\overline{P}(l|e)}\Tr_{E}\left[\left(\I_{A}\otimes R_{l|e}\right)\bigotimes_{i=1}^N\proj{\psi_{A_iE_i}}\right],
%\end{eqnarray}
%
%where $\overline{P}(l|e)$ is the probability that Eve observes the outcome $l$ after performing the $e-$th measurement 
satisfy
\begin{eqnarray}  \label{stprojstate}U_A\,\rho^{l|e}_{A}\,U_A^{\dagger}=\tilde{\rho}^{l|e}_{A'}\otimes\varrho_{A''}^{l|e},
\end{eqnarray}
where $\tilde{\rho}^{l|e}_{A'}$ are the reference post-measured states and $\varrho_{A''}^{l|e}$ are some auxiliary states and $U_A=U_{A_1}\otimes \ldots\otimes U_{A_N}$ where $U_{A_i}$ are the same 
unitary transformations as those in Eqs. (\ref{ststate}) and (\ref{stmea}), then we say that the post-measurement states are self-tested too. We refer to this method of self-testing as DI-certified remote preparation. Unlike standard self-testing, this approach allows self-testing probabilistically as the state to be certified occurs with the probability of obtaining the outcome $l$ by Eve $p(l|e)$. It should be noted here that the above definition of self-testing the sources (\ref{ststate}), measurements (\ref{stmea}), and the post-measurement states (\ref{stprojstate}) holds up to the complex conjugation of the reference states and measurements. The reason is that the correlations $\vec{p}$ are invariant under complex conjugation of quantum states and measurements.

Let us finally mention that throughout this work, we assume that the local states of $\ket{\psi_{A_iE_i}}$ are full-rank as the local measurements can only be certified on the supports of the local density matrices. 

\textit{Main results.---}We are now ready to present our scheme for self-testing of composite measurements and quantum states. It is divided into three major parts. The first one involves self-testing of the two-dimensional tomographically complete set of Pauli measurements in the measurements of the external parties and the two-qubit Bell states in the states generated by all the sources. The second one involves self-testing any quantum measurement performed by the central party using the states and external parties' measurements as certified in the above first part of the scheme. Lastly, the third part is concerned with applying the first two parts for self-testing the post-measurement state at the external parties' laboratories. As a result, our schemes enable self-testing any quantum state (pure or mixed) distributed between the external parties by Eve's measurements, while also facilitating the self-testing of quantum measurements performed by the central party.

\textit{Part 1.} To self-test the tomographically complete set of measurements in the external parties' measurement devices and the two-qubit singlets in the states generated by the sources, we introduce the following class of $2^N$ Bell inequalities that are suitable modifications of those introduced in Refs.
\cite{Flavio, sarkarPRL}:
\begin{widetext}
\begin{equation}\label{BE1Nm}
\mathcal{I}_{l}=(-1)^{l_1}\left[ (N-1)\left\langle\tilde{A}_{1,1}\prod_{i=2}^N A_{i,1} \right\rangle+\sum_{i=2}^N(-1)^{l_i}\Biggl\langle\tilde{A}_{1,0}A_{i,0}\Biggl\rangle-(-1)^{l_1}\sum_{i=2}^N (-1)^{l_i}\left\langle A_{1,2} A_{i,2} \prod_{\substack{j=2\\j\ne i}}^{N} A_{j,1}\right\rangle\right] \leqslant  \beta_C,
\end{equation}
\end{widetext}
\noindent where $l\equiv l_1\ldots l_N$ such that $l_1,l_2,\ldots,l_N=0,1$ is the binary representation of the outcome $l$, and
\begin{equation}\label{overAm}
    \tilde{A}_{1,0}=\frac{A_{1,0}-A_{1,1}}{\sqrt{2}},\qquad\tilde{A}_{1,1}=\frac{A_{1,0}+A_{1,1}}{\sqrt{2}}.
\end{equation}
Here, $A_{i,j}$ for $i=1,2,\ldots,N$ and $j=0,1,2$ are the observables measured by the external Alices.

Let us now briefly discuss the main properties of these Bell inequalities. First, it is direct to observe that their maximal values attainable using local deterministic strategies $\beta_C$, often referred to as the classical bound, are the same for any $l$ and amount to $\beta_C=(\sqrt{2}+1)(N-1)$ (see Fact 1 of \cite{SupMat} for a proof). Second, the maximal quantum values of $\mathcal{I}_l$, referred to as the quantum or Tsirelson's bounds, equal $3(N-1)$ and are achieved by the following observables of the external parties
\begin{eqnarray}\label{GHZObsm}
A_{1,0}&=&\frac{X+Z}{\sqrt{2}},\qquad A_{1,1}= \frac{X-Z}{\sqrt{2}},
\qquad A_{1,2}=Y, \nonumber\\
A_{i,0}&=&Z,\qquad A_{i,1}=X,\qquad A_{i,2}=Y
\end{eqnarray}
with $i=2,\ldots,N$ and the GHZ-like states

\begin{equation}\label{GHZvecsm}
\ket{\phi_l}=\frac{1}{\sqrt{2}}(\ket{l_1\ldots l_N}+(-1)^{l_{1}}|\overline{l}_1\ldots\overline{l}_N\rangle),
\end{equation}
where $l_1,l_2,\ldots,l_N=0,1$ and $\overline{l}_i=1-l_i$ (for a proof see Fact 2 in \cite{SupMat}). Notice that, the measurements in Eq. \eqref{GHZObsm} form a two-dimensional tomographically complete set of measurements for every party. We now comprise the result in the following theorem:

\begin{thm}\label{theo1m}
    Consider the scenario depicted in Fig. \ref{fig1}. If the Bell inequalities $\mathcal{I}_l$ \eqref{BE1Nm} for any $l$ are maximally violated when Eve chooses the input $e=0$ with each outcome occurring with probability $\overline{P}(l|0)=1/2^N$, then the measurements of the external parties are equivalent, according to Eq. \eqref{stmea}, to the reference measurements in \eqref{GHZObsm}, whereas the states $\ket{\psi_{A_iE_i}}$ are equivalent in the sense of Eq. (\ref{ststate}) to the two-qubit maximally entangled state $\ket{\phi^+}=(\ket{00}+\ket{11})/\sqrt{2}$.
\end{thm}

\noindent\textit{Sketch of the proof.} Here we provide a sketch of the proof, deferring its full version to Appendix A. First, Eve performs the measurement corresponding to the input $e=0$ and observes an outcome $l$, in which way she prepares a post-measurement state shared by the external parties. Now, observation of maximal violation of the Bell inequalities $\mathcal{I}_l$ by the external parties allows one to self-test the post-measurement state as well as the external parties' measurements. Then, the additional condition that the probabilities of obtaining the oucome $l$ by Eve when measuring $E_0=\{R_{l|0}\}$ is $1/2^N$ enables self-testing the states generated by the sources. This in turn also allows us to certify that Eve's measurement $E_0$ is the Greenberger-Horne-Zeilinger (GHZ) basis $\{\proj{\phi_l}\}$ where $\ket{\phi_l}$ are given in \eqref{GHZvecsm}.

\textit{Part 2.} Now, that the measurements performed by the external parties as well as the states generated by the sources are self-tested, we can show how the quantum network can be used to self-test any extremal POVM in the second measurement $E_1$ performed by the central party (see Fig. \ref{fig2}). Recall that a measurement is extremal if it cannot be expressed as a convex combination of other POVMs (cf. Ref. \cite{APP05} for a precise definition and a couple of criteria allowing one to decide whether a given measurement is extremal.)

For this purpose, let us assume that Eve's reference measurement $E_1'=\{R'_{l|1}\}$ is now some $K$-outcome extremal POVM defined on $\mathcal{H}_N=\left(\mathbbm{C}^2\right)^{\otimes N}$. %where the measurement elements are given by $R'_{l|1}=\lambda_l\proj{\phi_l}$ for some $0\leq \lambda_l\leq1$ and pure $N$-qubit pure states $\ket{\phi_l}$ such that $\sum_{l}\lambda_l\proj{\phi_l}=\mathbbm{1}$. 
Now, to self-test this reference measurement in $E_1$ the parties check whether for all outcomes $l=0,\ldots,K-1$ and sequences $i_1,\ldots,i_N$ with $i_j=0,1,2,3$, the correlations $\vec{p}$ satisfy the following conditions,
\begin{eqnarray}\label{betstatfull0}
\left\langle\tilde{A}_{1,i_1}\otimes\bigotimes_{k=2}^N A_{k,i_k}\otimes (R_{l|1})_E\right\rangle&=&f^l_{i_1,\ldots,i_N},
\end{eqnarray}
where $f^l_{i_1,\ldots,i_N}$ are real numbers coming from the decompositions of the measurement operators $R'_{l|1}$ in the  $N$-fold tensor products of the qubit Pauli matrices [see Appendix B of \cite{SupMat} for details]. We can thus state our next result.
\begin{thm}\label{theorem4m2}  Consider the scenario depicted in Fig. \ref{fig2} with states generated by the sources and measurements of the external parties being self-tested as in Theorem \ref{theo1m}. If the observed correlations $\vec{p}$ additionally satisfy the conditions \eqref{betstatfull0} for Eve's measurement $E_1=\{R_{l|1}\}$, then this measurement is self-tested to be the reference extremal POVM $E'_1$, according to the definition \eqref{stmea}. 
\end{thm}

\begin{figure}[t]
    \centering
    \includegraphics[width=0.8\linewidth]{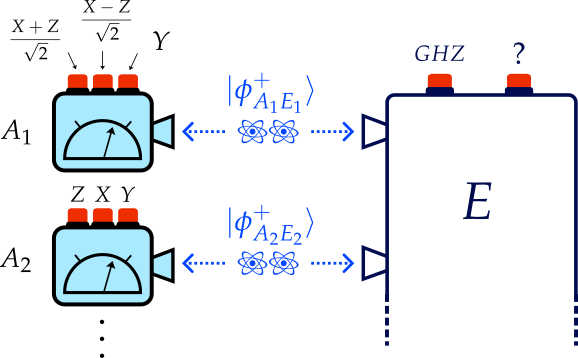}
    \caption{\textbf{Self-testing any quantum measurement.} First, we have self-tested all the measurements performed by each external party $A_i$ as \eqref{GHZObsm}, as well as, the sources as the maximally entangled two-qubit state $\ket{\phi_{A_iE_i}^+}$. Thus, Eve's first measurement is certified to the one in the GHZ basis \eqref{GHZvecsm}. These are the ingredients sufficient to certify Eve's arbitrary second measurement.}
    \label{fig2}
\end{figure}

The proof of this statement can be found in Appendix B of \cite{SupMat}.
Importantly, the above statement can be extended to any extremal measurement defined on a $D$-dimensional Hilbert space with arbitrary $D$.
In fact, one can always embed the measurement $\{R'_{l|1}\}$ into 
$\mathcal{H}_N=(\mathbbm{C}^2)^{\otimes N}$ where $N$ is the smallest natural number such that $D<2^N$. Then, in order to make the measurement fully supported on $\mathcal{H}_N$ one completes it with an additional projector $M^{\perp}$ such that $\sum_l R'_{l|1}+M^{\perp}=\mathbbm{1}_{2^N}$. This gives rise to an $(K+1)$-outcome extremal measurement to which Theorem \ref{theorem4m2} applies. Note that, in the case of composite measurements acting on $\bigotimes_{i}\mathbbm{C}^{D_i}$ each of the $i-$th local Hilbert spaces need to be embedded in $(\mathbbm{C}^2)^{\otimes n_i}$ such that $\sum_in_i=N$. 

%It should be noticed that similarly to the previous case, a rank-one extremal POVM $E_1'$ defined on $\mathbbm{C}^D$ for some $D$ can be embedded in $\mathcal{H}_N$ with $N$ being the minimal number such that $D<2^N$. To this end, one completes the measurement with additional $2^N-D$ rank-one projections $J_i$ so that the resulting measurement $\{R'_{1|1},\ldots,R'_{K|1},J_1,\ldots,J_{2^N-D}\}$ is also extremal (see Appendix C for more details).

We thus obtain a general way of certifying any extremal generalized measurement in the quantum networks scenario. Our scheme is based on violation of a class of simple Bell inequalities \eqref{BE1Nm} and some additional conditions \eqref{betstatfull0} that the observed correlations $\vec{p}$ must satisfy. It is however worth noticing that the generality of our scheme comes at the price of the complexity that grows fast with the dimension of the measurement under consideration and thus with the number of involved external parties. 

\textit{Indirect way to self-test non-extremal measurements.} Any non-extremal POVM $E$ can be expressed as a convex combination of extremal POVMs $E^{\mathrm{ex}}_e$, that is, $E=\sum_{e}p_e E^{\mathrm{ex}}_e$ (cf. Ref. \cite{APP05}). Consequently, a way to implement this non-extremal POVM in the star network in Fig. \ref{fig1} is to let Eve have $s$ inputs $e=1,\ldots,s$ corresponding to the measurements 
$E^{\mathrm{ex}}_e$, each chosen with probability $p_e$. Now, the above self-testing scheme for extremal measurements (Theorem \ref{theorem4m2}) can be used to self-test any of Eve's measurements corresponding to inputs $e=1,\ldots,s$ according to Definition \ref{stmea}. Consequently, one has an indirect way to self-test any non-extremal POVM in the star network.

\textit{Part 3.} Finally, we are going to show now that using the above setup one can self-test any quantum state. Consider again the network scenario depicted in Fig. \ref{fig2}. It is clear that Eve could remotely prepare different quantum states with the external parties by performing the concerned quantum measurement on the maximally entangled states generated by the sources. Thus, without loss of generality, the sources along with Eve can be considered a new preparation device. For instance, consider that Eve performs a projective measurement $E_1=\{\proj{e_1},\ldots,\proj{e_K}\}$. Then the post-measurement state shared by the external parties when she obtains an outcome $l$ is $\proj{e_l^*}$ where $*$ represents the complex conjugate. 
As a corollary to the previous statements, we can state the following fact.

\setcounter{thm}{0}
\begin{cor}
    Consider the scenario depicted in Fig. \ref{fig2}. If the states generated by the sources along with the measurements performed by the external parties and Eve maximally violate the Bell inequalities \eqref{BE1Nm} and satisfy the statistics \eqref{betstatfull0} when Eve's reference measurements are extremal, then the post-measurement states with the external parties can be self-tested according to \eqref{stprojstate}.
\end{cor}
As described below Eq. \eqref{stprojstate}, the state only occurs with a probability $p(l|e=1)$ and thus the success probability of self-testing the concerned state is $p(l|e=1)$.
The proof of this statement can be found in Appendix C of \cite{SupMat}. %It is important to emphasise here that, unlike the standard self-testing protocols, our scheme is not restricted to just pure entangled states. 
It is important to emphasise here that our approach of certified remote state preparation also allows device-independent certification of separable states and mixed states of any local dimension with the external parties, according to \eqref{ststate}.

It is simple to observe from the above Corollary that any pure state $\ket{\psi}$ can be self-tested with our scheme by appropriately choosing Eve's projective measurement such that one of its measurement elements, say $R'_{0|1}$
is exactly $R'_{0|1}=\proj{\psi}$.

Importantly, however, the scheme is not restricted to pure states only and can straightforwardly extended to mixed states. Consider a reference mixed state $\rho=\sum_kp_k\proj{\psi_k}$ acting on $\mathbbm{C}^d$ which we aim to self-test.
For this purpose, Eve needs to perform an extremal $3d-$outcome POVM in the Hilbert space $\mathbbm{C}^{2d}$ on her share of the joint state $\ket{\psi_{AE}}$ to create $\rho$ at the external parties' labs with the aid of post-processing. To construct this POVM, we first define two sets of $d$ mutually orthogonal vectors $\{\ket{\psi_k}\}$ and  $\{\ket{\phi_k}\}$ such that $\bra{\psi_k}\phi_{k'}\rangle=0$ for any $k,k'$. One then considers a pair $\{\ket{\psi_k},\ket{\phi_k}\}$ to construct a three-outcome trine POVM [cf. Ref. \cite{remik1}] as 
\begin{eqnarray}
    M_{k,1}&=&p_k \proj{\psi_k},\quad M_{k,2}=\alpha_k\proj{\tau_{k,2}},\nonumber\\ M_{k,3}&=&\beta_k\proj{\tau_{k,3}},
\end{eqnarray}
where for all $k$ the vectors $\ket{\tau_{k,2}}$ and $\ket{\tau_{k,3}}$ are certain linear combinations of the vectors $\ket{\psi_k}$ and $\ket{\phi_k}$ whose explicit forms are given in Eq. (164) in \cite{SupMat}
and
%
%with
%
%\begin{eqnarray}
    $\alpha_k=\beta_k=(2-p_k)/2$.
%\end{eqnarray}
%
The desired extremal $3d$-outcome rank-one POVM is one which is composed of all the trine POVM's, i.e., $E'_1=\{M_{k,m}\}_{k,m}$.

Now, when Eve performs this measurement on her share of the joint state and obtain the outcomes $l\equiv (k,1)$ with probability $p_k/2^N$, the post-measurement (normalized) states of the external parties are given by $\tilde{\rho}_{(k,1)}=\proj{\psi_k}$. These states are self-testable using Corollary 1. Then, the average state when Eve obtain the outcomes $(k,1)$ over all $k$ is thus certified to be
\begin{eqnarray}
    U_A\,\rho_{A}\,U^{\dagger}_A=\frac{1}{\sum_kp(k,1)}\sum_{k}p(k,1) U_A\,\rho_{A}^{(k,1)}\,U_A^{\dagger}\nonumber\\=\left(\sum_{k}M_{(k,1)}\right)\otimes \tilde{\rho}_{A''}=\rho_{A'}\otimes \tilde{\rho}_{A''},
\end{eqnarray}
or its complex conjugate, 
which is exactly what we promised. Let us  say that the external parties want to perform a task with the state $\rho$. Then, they perform their task with all the post-measurement states of Eve and after the task is done, just like Bell's scenario Eve announces the outcomes. Now all the parties only consider the statistics when Eve obtains $(k,1)$ and discard the rest of the statistics. 

In practice, it is never possible to obtain the exact statistics required for self-testing of states and measurements. In Appendix D, we analyse the effect of noise on the self-testing correlations. In particular, we show in Theorem 4 that if Eve's actual measurement corresponding to $e=1$ in the device deviates from the ideal one, with some error $\varepsilon$, then the observed correlations will be close to the ideal ones with an error which is a function of $\sqrt{\varepsilon}$ and the number of external parties $N$. For our analysis, we do not impose any restrictions on the states generated by the source, the measurements of the external parties or the noise model. Consequently, the derived robustness statement serves as an upper bound on the error in the correlations, which can be further improved by using numerical techniques or particular noise models. For Eve's measurement with input $e=0$ the proof in \cite{sarkarPRL} for $N=3$ can be straightforwardly adapted to arbitrary $N$, allowing one to show that if 
the observed values of the functionals $\mathcal{I}_l$ are $\varepsilon$-close to the 
optimal values, then the measurement is $\sqrt{2\varepsilon}$ close to the ideal GHZ-state measurements.

\textit{Discussions.---}
After completing our work, we have identified several interesting follow-up questions. %While we have provided a proof-of-principle scheme for self-testing any quantum state and extremal measurement, it is important to note that in actual experiments, the exact statistics mentioned may not be attainable, and instead, values that are close to them will be obtained. 
First of all, our certification scheme assumes that the sources are independent and therefore one crucial follow-up question is whether the scheme can be generalised to the case when the sources are not independent and exhibit either classical or quantum correlations. A possible approach is to give Eve additional inputs, which could then be used to self-test each source, ensuring it produces a maximally entangled state up to some some ancillary junk state. Thus, the sources would be independent on the support of the maximally entangled states. However, the states could still be entangled through the junk part, which would require more advanced techniques beyond the scope of this work. Additionally, our current scheme requires highly entangled states and measurements, which can be costly to generate. Therefore, it would be interesting to explore the possibility of self-testing quantum states and measurements using partially entangled states and measurements. A possible way is to generalise the idea of self-testing the GHZ basis stated above to other measurements is through finding Bell inequalities that are maximally violated by all post-measured states corresponding to the measurement elements in the other basis. A simpler problem in this regard will be to consider partially entangled states in the network and then find Bell inequalities that are maximally violated by the corresponding post-measured states, which will again be some partially entangled states. The simplest star network, which is composed of two external parties and a central party, has been implemented in \cite{netexp1, netexp2, netexp3} on optical platforms. Our scheme requires one to extend this setup to the scenario with a larger number of external parties. However, we recognise that performing a GHZ basis measurement is challenging to implement in an optical setup. Thus, it will be highly relevant to provide alternate schemes that do not require GHZ measurements.
%Lastly, as we have utilised the framework of quantum networks, it would be valuable to investigate if our approach can be utilised for multiparty quantum cryptography or key distribution.

\begin{center}
    \textbf{Acknowledgements}
\end{center}
 This project was funded within the QuantERA II Programme (VERIqTAS project) that has received funding from the European Union’s Horizon 2020 research and innovation programme under Grant Agreement No 101017733 and from the Polish National Science Center (project No 2021/03/Y/ST2/00175). SS also acknowledges the National Science Center, Poland, grant Opus 25, 2023/49/B/ST2/02468. AO acknowledges the Polish National Science Center (Grant No. 2022/46/E/ST2/00115).

\iffalse
\begin{center}
    \textbf{Data Availability}
\end{center}
No data set was generated or analyzed in the current study.

\begin{center}
    \textbf{Code Availability}
\end{center}
No codes were generated or used in the current study.

\begin{center}
    \textbf{Author contributions}
\end{center}
S.S. and R.A. conceived the idea. S.S., A.C.O., and R.A. designed the proofs and prepared the manuscript.

\begin{center}
    \textbf{Competing interests}
\end{center}
The authors declare no competing interests.
\fi
%\textit{Acknowledgements.--}

%\bibliography{ref.bib}
\providecommand{\noopsort}[1]{}\providecommand{\singleletter}[1]{#1}%
%

%%%%%%%%%%%%%%%%%%%%%%%%%%%%%%%%

%\iifalse

%%%%%%%%%%%%%%%%%%%%%%%%%%%%%%%%%%%%%%%%%%%%%%%%%%%%%%%%%%%%%%%%%%%
%\maketitle

\onecolumngrid

\appendix

\setcounter{thm}{0}
\section{Appendix A: Self-testing of N-GHZ bases, measurements of external parties, and the states prepared by the preparation devices}
\subsection{Bell inequalities}

%In this section, we provide proof of self-testing the N-GHZ basis along with the measurements of the external parties and the states distributed by the sources. For this purpose, 
Let us consider the following $2^N$ Bell inequalities inspired from \cite{Flavio, sarkarPRL}:
\begin{eqnarray}\label{BE1N}
\mathcal{I}_{l}=(-1)^{l_1} \left\langle(N-1)\tilde{A}_{1,1}\otimes\bigotimes_{i=2}^N A_{i,1} +\sum_{i=2}^N(-1)^{l_i}\tilde{A}_{1,0}\otimes A_{i,0}-(-1)^{l_1}\sum_{i=2}^N (-1)^{l_i}A_{1,2}\otimes A_{i,2}\otimes \bigotimes_{\substack{j=2\\j\ne i}}^{N} A_{j,1} \right\rangle\leqslant  \beta_C,
\end{eqnarray}
where $l\equiv l_1l_2\ldots l_N$ with $l_i=0,1$ for each $i=1,\ldots,N$ is the binary representation of $l$, and
\begin{eqnarray}\label{overA}
    \tilde{A}_{1,0}=\frac{A_{1,0}-A_{1,1}}{\sqrt{2}},\qquad\tilde{A}_{1,1}=\frac{A_{1,0}+A_{1,1}}{\sqrt{2}}.
\end{eqnarray}
Here, $A_{i,j}$ for $i=1,2,\ldots,N$ and $j=0,1,2$ correspond to the observables of all the parties.

\begin{fakt}\label{classfact}
    The maximal classical value of the Bell expressions $\mathcal{I}_l$ for any $l$ is $\beta_C= (\sqrt{2}+1)(N-1)$.
\end{fakt}
\begin{proof}To compute the maximal classical value of $\mathcal{I}_l$ it is enough to optimize 
it over local deterministic strategies for which all the expectation values are products of the local ones, and, moreover, the local expectation values equal either $1$ or $-1$. Let us begin by noting that the maximal value of the last sum in Eq. (\ref{BE1N}) over such strategies is $N-1$ and it is independent of the choice of the maximal values of the first two term. To determine the latter we can
to consider two cases: $A_{1,0}=A_{1,1}$ and $A_{1,0}\neq A_{1,1}$. In the first one, the maximal value of the first term in Eq. (\ref{BE1N}) is $\sqrt{2}(N-1)$, whereas the second term vanishes. In the second case, the first term in $\mathcal{I}_l$ vanishes, whereas the maximal value of the second term
is again $\sqrt{2}(N-1)$. Thus, $\beta_C=(\sqrt{2}+1)(N-1)$.
%
%Let us first notice that the classical value of the Bell inequality composed of the first two terms in %the Bell expression \eqref{BE1N} when $l_i=0$ was found in \cite{Flavio} to be $\sqrt{2}(N-1)$. Now, %the final term in \eqref{BE1N} is algebraically bounded by $N-1$. Thus, the classical value of the Bell %expressions \eqref{BE1N} for any $l$ can be at most $(\sqrt{2}+1)(N-1)$.  
\end{proof}
\begin{fakt}The maximal quantum value of the Bell expression $\mathcal{I}_l$ is $3(N-1)$ and it is achieved by the following observables 
\begin{eqnarray}\label{GHZObs}
A_{1,0}=\frac{X+Z}{\sqrt{2}},\quad A_{1,1}= \frac{X-Z}{\sqrt{2}},\quad A_{1,2}=Y, \quad A_{i,0}=Z,\quad A_{i,1}=X,\quad A_{i,2}=Y \quad(i=2,\ldots,N)
\end{eqnarray}
as well as the GHZ-like states 
\begin{equation}\label{GHZvecs}
\ket{\phi_l}=\frac{1}{\sqrt{2}}(\ket{l_1\ldots l_N}+(-1)^{l_{1}}|\overline{l}_1\ldots\overline{l}_N\rangle),
\end{equation}
where $l\equiv l_1l_2\ldots l_N$ with $l_1,l_2,\ldots,l_N=0,1$ and $\overline{l}_i=1-l_i$. 
\end{fakt}
\begin{proof}
The following proof is inspired by Refs. \cite{Flavio, sarkarPRL}. Let us first associate a Bell operator to each of the Bell expressions $\mathcal{I}_l$ \eqref{BE1N} of the following form
\begin{equation}\label{BE1Nop}
\mathcal{\hat{I}}_{l}=(N-1)(-1)^{l_1}\tilde{A}_{1,1}\otimes\bigotimes_{i=2}^N A_{i,1}
%-(-1)^{l_1}\sum_{\substack{j_i=1,2\\ \sum_{i=2}^N{j_i}=N+1}}\tilde{A}_{1,1}\bigotimes_{i=2}^NA_{i,j_i}
+\sum_{i=2}^N(-1)^{l_1+l_i}\tilde{A}_{1,0}\otimes A_{i,0}-\sum_{i=2}^{N}(-1)^{l_i} A_{1,2}\otimes A_{i,2}\otimes\bigotimes_{\substack{j=2\\j\ne i}}^{N} A_{j,1}.
\end{equation}
Now, for each $\mathcal{\hat{I}}_l$ we can find the following sum-of-squares (SOS) decomposition,
\begin{eqnarray}\label{SOS1}
   2\left[3(N-1)\,\I-\mathcal{\hat{I}}_{l}\right]\geq(N-1)P_{l_1}^2+\sum_{i=2}^NR_{i,l_1,l_i}^2+\sum_{i=2}^NQ_{i,l_i}^2, 
\end{eqnarray}
where 

\begin{subequations}\label{SOS2}
\begin{equation}\label{SOS2_1}
    P_{l_1}=(-1)^{l_1}\tilde{A}_{1,1}-\bigotimes_{i=2}^N A_{i,1},
\end{equation}
%\begin{equation}\label{SOS2_3}
 %  P_{3,l_1,j_1,\ldots,j_N}=(-1)^{l_1}A_{1,2}-\bigotimes_{i=2}^{N} A_{i,j_i} 
%\end{equation}
\begin{equation}\label{SOS2_4}
    R_{i,l_1,l_i}=(-1)^{l_1+l_i}\tilde{A}_{1,0}-A_{i,0},
\end{equation}
\begin{equation}\label{SOS2_2}
   Q_{i,l_i}=(-1)^{l_i}A_{1,2}+A_{i,2}\otimes\bigotimes_{\substack{j=2\\j\ne i}}^{N} A_{j,1}.
\end{equation}
\end{subequations}
Notice that by expanding the terms on the right-hand side of the above decomposition \eqref{SOS1} one obtains
\begin{eqnarray}\label{SOS11}
 (N-1)\left(\tilde{A}_{1,1}^2+\tilde{A}_{1,0}^2+A_{1,2}^2\right)+(N-1)\bigotimes_{i=2}^N A_{i,1}^2+\sum_{i=2}^{N} A_{i,2}^2\otimes\bigotimes_{\substack{j=2\\j\ne i}}^{N} A_{j,1}^2+\sum_{i=2}^NA_{i,0}^2-2\ \mathcal{\hat{I}}_{l}\nonumber\\ \quad\leqslant \  2\left[3(N-1)\ \I-\ \mathcal{\hat{I}}_{l}\right],
\end{eqnarray}
where to arrive at the last line we used the fact that $\tilde{A}_{i,j}^2\leqslant \I$. Thus,
one can conclude from \eqref{SOS11} that $3(N-1)$ is an upper bound on $\mathcal{I}_{l}$ when using quantum theory. One can easily verify that the Bell inequalities \eqref{BE1N} attain the value $3(N-1)$ when the states shared among the parties are the GHZ-like state $\ket{\phi_l}$ \eqref{GHZvecs} and the observables given in Eq. (\ref{GHZObs}). Thus, the Tsirelson's bound of the Bell inequalities \eqref{BE1N} for any $l$ is $3(N-1)$.
%Let us show that this upper bound is in fact the quantum bound of the Bell inequalities \eqref{BE1N} for any $l$
\end{proof}

An important observation from the above SOS decomposition is that when the maximal violation of the above Bell inequalities is attained, the right-hand side of Eq. \eqref{SOS1} is $0$. Thus, any state $\ket{\psi}$ and observables $A_{i,j}$ that attain the maximal violation of the Bell inequalities \eqref{BE1N} for any $l$ is given by
    \begin{subequations}
    \begin{equation}\label{SOSrel1}
          (-1)^{l_1}\tilde{A}_{1,1}\ket{\psi}=\bigotimes_{i=2}^N A_{i,1}\ket{\psi},
    \end{equation}
        % \begin{equation}\label{SOSrel3}
         %     (-1)^{l_1}A_{1,2}\ket{\psi}=\bigotimes_{i=2}^NA_{i,j_i} \ket{\psi}\qquad j_i=1,2\quad s.t.\quad \sum_{i=2}^Nj_i=N
        %\end{equation}
         \begin{equation}\label{SOSrel4}
              (-1)^{l_1+l_i}\tilde{A}_{1,0}\ket{\psi}=A_{i,0} \ket{\psi}\qquad (i=2,\ldots,N),
        \end{equation}
         \begin{equation}\label{SOSrel2}
             % (-1)^{l_1}\tilde{A}_{1,1}\ket{\psi}=\bigotimes_{i=2}^NA_{i,j_i} \ket{\psi}\qquad j_i=1,2\quad s.t.\quad \sum_{i=2}^Nj_i=N+1
             -(-1)^{l_i}A_{1,2}\ket{\psi}=A_{i,2}\otimes\bigotimes_{\substack{j=2\\j\ne i}}^{N} A_{j,1}\ket{\psi}\qquad (i=2,\ldots,N).
        \end{equation}
    \end{subequations}
Also, one can conclude from Eq. \eqref{SOS11} that the observables that attain the maximum violation of the Bell inequalities \eqref{BE1N} must be unitary, that is, $A_{i,j}^2=\I$. These relations will be particularly useful for self-testing.

\subsection{Self-testing}
\begin{thm}\label{theorem1}
Assume that the Bell inequalities \eqref{BE1N} for any $l$ are maximally violated and each outcome of the central party occurs with probability $\overline{P}(l|0)=1/2^N$ where $N$ denotes the number of external parties. The sources $P_i$ prepare the states $\ket{\psi_{A_iE_i}}\in \mathcal{H}_{A_i}\otimes\mathcal{H}_{E_i}$ for $i=1,2,\ldots,N$, the measurement with the central party is given by $\{R_{l|0}\}$ for $l=l_1l_2\ldots l_N$ such that $l_i=0,1$ which acts on $\bigotimes_{i=1}^N\mathcal{H}_{E_i}$, the  local observables for each of the party is given by $A_{i,0}$, $A_{i,1}$, and $A_{i,2}$ for $i=1,2,\ldots,N$ which act on $\mathcal{H}_{A_i}$. Then,
\begin{enumerate}
    \item The Hilbert spaces of all the parties decompose as $\mathcal{H}_{A_i}=\mathcal{H}_{A_i'}\otimes\mathcal{H}_{A_i''}$ and $\mathcal{H}_{E_i}=\mathcal{H}_{E_i'}\otimes\mathcal{H}_{E_i''}$, where $\mathcal{H}_{A_i'}$ and 
    $\mathcal{H}_{E_i'}$ are qubit Hilbert spaces, whereas $\mathcal{H}_{A_i''}$ and $\mathcal{H}_{E_i''}$
    are some finite-dimensional but unknown auxiliary Hilbert spaces.
    \item There exist local unitary transformations $U_{A_i}:\mathcal{H}_{A_i}\rightarrow(\mathbb{C}^2)_{A'}\otimes\mathcal{H}_{A''_i}$ and $U_{E_i}:\mathcal{H}_{E_i}\rightarrow(\mathbb{C}^2)_{A'}\otimes \mathcal{H}_{E''_i}$ %such as ,
    such that the states are given by
\begin{eqnarray}\label{A10}
U_{A_i}\otimes U_{E_i}\ket{\psi_{A_iE_i}}=\ket{\phi^+_{A_i'E'_i}}\otimes\ket{\xi_{A_i''E''_i}}.
\end{eqnarray}
\item The measurement of the central party is
\begin{eqnarray}\label{statest1}
 U_ER_lU_E^{\dagger} =\proj{\phi_l}_{E'}\otimes\I_{E''}\qquad \forall l,
\end{eqnarray}
where $U_E=U_{E_1}\otimes\ldots\otimes U_{E_N}$ and $\ket{\phi_l}$ are given in Eq. \eqref{GHZvecs}. The measurements of all the other parties are given by
\begin{eqnarray}\label{mea1}
U_{A_1}A_{1,0}\,U_{A_1}^{\dagger}&=&\left(\frac{X+Z}{\sqrt{2}}\right)_{A_1'}\otimes\I_{A_1''}, \quad U_{A_1}A_{1,1}\,U_{A_1}^{\dagger}=\left(\frac{X-Z}{\sqrt{2}}\right)_{A_1'}\otimes\I_{A_1''},\nonumber\\
U_{A_i}A_{i,0}\,U_{A_i}^{\dagger}&=&Z_{A_i'}\otimes\I_{A_i''},\quad\quad\ \ \ \ \  U_{A_i}A_{i,1}\,U_{A_i}^{\dagger}=X_{A_i'}\otimes\I_{A_i''}\qquad (i=2,\ldots,N)\label{stmea1}
\end{eqnarray}
and,
\begin{eqnarray}\label{mea2}
    U_{A_i}A_{i,2}\,U_{A_i}^{\dagger}&=&\pm Y_{A_i'}\otimes\I_{A_i''} \qquad\forall i.\label{stmea2}
\end{eqnarray}
\end{enumerate}
\end{thm}
\begin{proof} Before proceeding, let us recall that the state shared among the external parties $A_i$ when the central party $E_i$ gets an outcome $l$ is given by
\begin{eqnarray}\label{54}
\rho^l_{{A}}=\frac{1}{\overline{P}(l|0)}\Tr_{E}\left[\left(\I_{A}\otimes R_l\right)\bigotimes_{i=1}^N\proj{\psi_{A_iE_i}}\right],
\end{eqnarray}
where we denoted $A_1\ldots A_N\equiv A$ and $E_1\ldots E_N\equiv E$. Similar notation will be used for $A'$, $A''$, $E'$, and $E''$. Now, we first find the states $\rho^l_{{A}}$ for any $l$ and characterise the measurements $A_{i,j}$ for any $i$ and $j=0,1$. Then, using these derived states, we certify the states generated by the sources $\ket{\psi_{A_iE_i}}$ for all $i$. Then, using both of these results we certify the measurement $\{R_{l|0}\}$. For simplicity, we represent $R_{l|0}\equiv R_l$ in the below proof. Finally, using all the derived results, we characterise the observables $A_{i,2}$.\\
\begin{center}
    \textbf{A.\ \  Post-measurement states $\rho^l_{{A}}$}
\end{center}
To characterize the post-measurement states $\rho^l_{{A}}$ based on violation of our Bell inequalities we need to consider the relations \eqref{SOSrel1} and \eqref{SOSrel4} obtained from the sum of squares decomposition of the corresponding Bell operators. First, let us consider a purification of the state $\rho^l_{{A}}$ which is obtained by adding an ancillary system $G$ such that 
\begin{eqnarray}
   \rho^l_{{A}}=\Tr_G\left(\proj{\psi_l}_{{AG}}\right).
\end{eqnarray}
In what follows, for simplicity, we drop the subscript $AG$ from the state $\ket{\psi_l}_{AG}$. 

Now, for each outcome of the measurement performed by the central party, the corresponding Bell inequalities \eqref{BE1N} are maximally violated by the post-measurement states $\ket{\psi_l}$. Thus, these states and the observables of every party must satisfy the relations \eqref{SOSrel1} and \eqref{SOSrel4}, which we state as:
\begin{subequations}
    \begin{align}
    (-1)^{l_1}\frac{A_{1,0}+A_{1,1}}{\sqrt{2}}\ket{\psi_l} &=\bigotimes_{i=2}^N A_{i,1}\ket{\psi_l} \label{SOSrel1p},\\
    (-1)^{l_1+l_i}\frac{A_{1,0}-A_{1,1}}{\sqrt{2}}\ket{\psi_l} &=A_{i,0}\ket{\psi_l} \label{SOSrel2p}
    \end{align}
\end{subequations}
for any $l=l_1l_2\ldots l_N$ and $i=2,3,\ldots,N$. Notice that any measurement can only be characterised on the support of the local states of $\ket{\psi_l}$. Thus, let us denote by $\Pi_l^{A_1}$ the projection onto the support of $\rho_{A_1}^l=\Tr_{A\setminus A_1}(\rho_A^l)$ and apply it to the above equations to obtain:
\begin{subequations}
    \begin{align}
    (-1)^{l_1}\frac{\mathbb{A}_{1,0}^{(l)}+\mathbb{A}_{1,1}^{(l)}}{\sqrt{2}}\ket{\psi_l} &=\bigotimes_{i=2}^N A_{i,1}\ket{\psi_l},\label{SOS1proj}\\
    (-1)^{l_1+l_i}\frac{\mathbb{A}_{1,0}^{(l)}-\mathbb{A}_{1,1}^{(l)}}{\sqrt{2}}\ket{\psi_l} &=A_{i,0}\ket{\psi_l}\qquad (i=2,3,\ldots,N),\label{SOS2proj}
    \end{align}
\end{subequations}
where we introduce the notation $\mathbb{A}_{1,j}^{(l)}=\Pi_l^{A_1} A_{1,j} \Pi_l^{A_1}$ for $j\in\{0,1\}$. Observe that $\Pi_l^{A_1}\ket{\psi_l}=\ket{\psi_l}$.

Let us now recall that to achieve the maximal violation of the Bell inequalities \eqref{BE1N} the observables must be unitary, that is, $A_{i,j}^2=\I$. Now, consider the relations \eqref{SOS1proj} and \eqref{SOS2proj}, and multiply them by $\bigotimes_{i=2}^N A_{i,1}$ and $A_{i,0}$, respectively. By using the same relations on themselves again, we obtain that:
\begin{subequations}
\begin{align}
    \left(\mathbb{A}_{1,0}^{(l)}+\mathbb{A}_{1,1}^{(l)}\right)^2\ket{\psi_l} &= 2\ket{\psi_l}; \label{SOSrel3p}\\
    \left(\mathbb{A}_{1,0}^{(l)}-\mathbb{A}_{1,1}^{(l)}\right)^2\ket{\psi_l} &= 2\ket{\psi_l}. \label{SOS2rel3p}
\end{align}
\end{subequations}
We can assume that the local states $\rho_{A_1}$ are full-rank and thus invertible. Therefore, if we take a partial trace over all the other subsystems except the first one, we obtain from \eqref{SOSrel3p} and \eqref{SOS2rel3p} that
\be
   \left(\mathbb{A}_{1,0}^{(l)}+\mathbb{A}_{1,1}^{(l)}\right)^2=2\Pi_l^{A_1}
\ee
and
\be
   \left(\mathbb{A}_{1,0}^{(l)}-\mathbb{A}_{1,1}^{(l)}\right)^2=2\Pi_l^{A_1},
\ee
which results in
\be\label{antico_proj}
\left\{\mathbb{A}_{1,0}^{(l)},\mathbb{A}_{1,1}^{(l)}\right\}=0.
\ee
Also, because $\left(\mathbb{A}_{1,i}^{(l)}\right)^2\leqslant \Pi_l^{A_1}$, the above equations imply that $\left(\mathbb{A}_{1,i}^{(l)}\right)^2=\Pi_l^{A_1}$ for $i\in\{0,1\}$. This means 
that $\mathbb{A}_{1,i}^{(l)}$ are unitary on the support of $\rho_{A_1}^l$ and and a consequence, 
the observables $A_{1,i}$ can be represented as direct sums,
\begin{equation}
    A_{1,i}=\mathbb{A}_{1,i}^{(l)}\oplus F_{i}^{(l)},
\end{equation}
where $F_{i}^{(l)}$ act on the complement of $\mathrm{supp}(\rho_{A_1}^{l})$ in $\mathcal{H}_{A_1}$.
This means that the anticommutation relations \eqref{antico_proj} can also be stated as
\begin{equation}
\left\{\mathbb{A}_{1,0}^{(l)},\mathbb{A}_{1,1}^{(l)}\right\}= \Pi_l^{A_1}\left\{A_{1,0},A_{1,1}\right\}=0,
\end{equation}
which allows us to conclude that
\begin{equation}
    \sum_{l}\Pi_l^{A_1}\left\{A_{1,0},A_{1,1}\right\}=0.
\end{equation}

Let us now show that $\sum_l \Pi_l^{A_1}>0$, that is, the sum of all projections $\Pi_{l}^{A_1}$ is invertible. To see this, suppose that $\sum_l \Pi_l^{A_1}\ngtr 0$. That implies that exists $\ket{\phi}$ such that $\sum_l \Pi_l \ket{\phi}=0$. Thus, for every $l$, we must have $\Pi_l\ket{\phi}=0$, which means that $\ket{\phi}$ belongs to the kernel of $\rho_{A_1}^{l}$, i.e., $\rho_{A_1}^{l}\ket{\phi}=0$. By summing over $l$, we obtain $\rho_{A_1}\ket{\phi}=0$, which contradicts the assumption that the local states of $\ket{\psi_{A_iE_i}}$ are full rank. Therefore, we must have
\be\label{sum_Pi_l}
\sum_l \Pi_l^{A_1}>0,
\ee
which implies that 
%Finally, from Eqs. \eqref{antico_proj} and \eqref{sum_Pi_l}, and because both $A_{1,i}$ and %%$\mathbb{A}_{1,i}$ ($i\in\{0,1\}$) are unitary on the respective Hilbert spaces, we can apply Lemma 2 %stated in the Appendices of Ref. \cite{sarkarPRL} to conclude that \textbf{(Remik: Maybe we can %repeat this lemma here?)}
%
\be
\{A_{1,0},A_{1,1}\}=0.
\ee

As proven in e.g. Ref. \cite{Jed1}, for a pair of unitary observables with eigenvalues $\pm 1$ that anti-commute, there exist a local unitary $V_{A_1}:\mathcal{H}_{A_1}\rightarrow\mathcal{H}_{A_1}$ such that 
\begin{eqnarray}
  V_{A_1}A_{1,0}V_{A_1}^{\dagger}&=&Z_{A_1'}\otimes\I_{A_1''},\qquad  V_{A_1}A_{1,1}V_{A_1}^{\dagger}=X_{A_1'}\otimes\I_{A_1''}.
\end{eqnarray}
Now, we can always apply a unitary $V':\mathcal{H}_{A_1'}\rightarrow\mathcal{H}_{A_1'}$ such that $V'ZV'^{\dagger}=(X+Z)/\sqrt{2}$ and $V'XV'^{\dagger}=(X-Z)/\sqrt{2}$. Thus, denoting $U_{A_1}=(V'\otimes\mathbbm{1}_{A_1''})\,V_{A_1}$, we obtain that there exists a unitary operation $U_{A_1}:\mathcal{H}_{A1}\to (\mathbbm{C}^2)_{A_1'}\otimes \mathcal{H}_{A_1''}$ for which,
\begin{eqnarray}\label{A1}
U_{A_1}A_{1,0}\,U_{A_1}^{\dagger}&=&\frac{X+Z}{\sqrt{2}}\otimes\I_{A_1''}, \qquad U_{A_1}A_{1,1}\,U_{A_1}^\dagger=\frac{X-Z}{\sqrt{2}}\otimes\I_{A_1''}.
\end{eqnarray}

Now, let us derive the form of the observables $A_{i,0}$ and $A_{i,1}$ of all the other parties. For this purpose, let us first consider the $k$-th party and then use the derived observables \eqref{A1} to rewrite the relations \eqref{SOSrel1p} and \eqref{SOSrel2p} as
\begin{eqnarray}\label{SOSrel5p}
   (-1)^{l_1}X_{A_1'} \ket{\psi'_l}=\mathbb{A}_{k,1}^{(l)}\otimes\bigotimes_{i=2\setminus\{k\}}^N A_{i,1}\ket{\psi'_l}
\end{eqnarray}
and,
\begin{eqnarray}\label{SOSrel6p}
   (-1)^{l_1+l_k}Z_{A_1'}\ket{\psi'_l}=\mathbb{A}_{k,0}^{(l)}\ket{\psi'_l},
\end{eqnarray}
where $\ket{\psi'_l}=U_{A_1}\ket{\psi_l}$, and $\mathbb{A}_{k,i}^{(l)}=\Pi_{l}^{A_k}A_{k,i}\, \Pi_{l}^{A_k}$ for $i\in\{0,1\}$, is the projected operator defined in analogous way as before. Now, multiplying \eqref{SOSrel5p} with $Z_{A_1'}$ and then using \eqref{SOSrel6p} on the right-hand side of the obtained expression, we get
\begin{eqnarray}\label{SOSrel7p}
    (-1)^{l_k}(ZX)_{A_1'} \ket{\psi'_l}=\mathbb{A}_{k,1}^{(l)}\mathbb{A}_{k,0}^{(l)}\otimes\bigotimes_{i=2\setminus\{k\}}^N A_{i,1}\ket{\psi'_l}.
\end{eqnarray}
Now, multiplying \eqref{SOSrel6p} with $X_{A_1'}$ and then using \eqref{SOSrel5p} on the right-hand side of the obtained expression, we get
\begin{eqnarray}\label{SOSrel8p}
    (-1)^{l_k}(XZ)_{A_1'} \ket{\psi'_l}=\mathbb{A}_{k,0}^{(l)}\mathbb{A}_{k,1}^{(l)}\otimes\bigotimes_{i=2\setminus\{k\}}^N A_{i,1}\ket{\psi'_l}.
\end{eqnarray}
Adding \eqref{SOSrel7p} and \eqref{SOSrel8p} and using the fact that the Pauli matrices satisfy $ZX+XZ=0$, we finally obtain that
\begin{eqnarray}
 \left \{\mathbb{A}_{k,0}^{(l)},\mathbb{A}_{k,1}^{(l)}\right\}\otimes\bigotimes_{i=2\setminus\{k\}}^N A_{i,1}\ket{\psi'_l}=0.
\end{eqnarray}
Now, using the fact that the observables $A_{i,j}$ are unitary, we get that
\begin{eqnarray}
   \left\{\mathbb{A}_{k,0}^{(l)},\mathbb{A}_{k,1}^{(l)}\right\}=0\qquad (k=2,\ldots,N).
\end{eqnarray}
By applying $X_{A_1'}$ and $Z_{A_1'}$ to Eqs. \eqref{SOSrel5p} and \eqref{SOSrel6p}, respectively, we obtain that $(\mathbb{A}_{k,i}^{(l)})^2=\Pi_l^{A_k}$. Analogously to the previous argument, we can conclude that
\be
\left\{A_{k,0},A_{k,1}\right\}=0\qquad (k=2,\ldots,N).
\ee
Again, using the result presented in the appendix of \cite{Jed1}, we observe that there exists unitaries $U_{A_i}:\mathcal{H}_{A_i}\rightarrow\mathcal{H}_{A_i}$ such that
\begin{eqnarray}\label{An}
   U_{A_k}A_{k,0}\,U_{A_k}^{\dagger}=Z_{A_i'}\otimes\I_{A_i''},\qquad  U_{A_k}A_{k,1}\,U_{A_k}^{\dagger}=X_{A_i'}\otimes\I_{A_i''}\qquad (k=2,\ldots,N).
\end{eqnarray}

Now, let us characterise the states $\rho^l_{{A}}$ for any $l$. For this purpose, we again consider the relations \eqref{SOSrel1} and \eqref{SOSrel4} and substitute into them the derived observables from \eqref{A1} and \eqref{An}. This gives us
\begin{eqnarray}\label{SOSrel10}
   (-1)^{l_1} \bigotimes_{i=1}^N X_{A_i'}\ket{\tilde{\psi}_l}=\ket{\tilde{\psi}_l}
\end{eqnarray}
and,
\begin{eqnarray}\label{SOSrel11}
   (-1)^{l_1+l_i}Z_{A_1'}\otimes Z_{A_i'}\ket{\tilde{\psi}_l}=\ket{\tilde{\psi}_l}\qquad (i=2,\ldots,N),
\end{eqnarray}
where $\ket{\tilde{\psi}_l}=(U_{A_1}\otimes\ldots\otimes U_{A_N} )\ket{\psi_l}$. As concluded above, the local Hilbert spaces $\mathcal{H}_{A_i}$ are even-dimensional, and therefore any state $\ket{\tilde{\psi}_l}$ belonging to the Hilbert space $\bigotimes_{i=1}^N\mathcal{H}_{A_i}\otimes\mathcal{H}_G$ can be written as
\begin{eqnarray}\label{genstate4}
   \ket{\tilde{\psi}_l}=\sum_{i_1,i_2,\ldots i_N=0,1}\ket{i_1i_2\ldots i_N}_{A'}\ket{\phi^l_{i_1i_2\ldots i_N}}_{A''G},
\end{eqnarray}
where $\ket{\phi^l_{i_1i_2\ldots i_N}}_{A''G}$ are some unknown (in general unnormalized) 
states corresponding to the auxiliary degrees of freedom $A_i''$ and the system $G$
which is responsible for purification.

Putting the state \eqref{genstate4} into the conditions \eqref{SOSrel11}, we obtain
\begin{eqnarray}
   (-1)^{l_1+l_k}\sum_{i_1,i_2,\ldots i_N=0,1}(-1)^{i_1+i_k}\ket{i_1i_2\ldots i_N}\ket{\phi^l_{i_1i_2\ldots i_N}}=\sum_{i_1,i_2,\ldots i_N=0,1}\ket{i_1i_2\ldots i_N}\ket{\phi^l_{i_1i_2\ldots i_N}},
\end{eqnarray}
where $\ket{\phi^l_{i_1i_2\ldots i_N}}$ are in general unnormalised. Projecting the above formula on $\bra{i_1i_2\ldots i_N}$, we get
\begin{eqnarray}\label{ststate2}
   (-1)^{l_1+l_k}(-1)^{i_1+i_k}\ket{\phi^l_{i_1i_2\ldots i_N}}=\ket{\phi^l_{i_1i_2\ldots i_N}}\qquad (k=2,\ldots,N).
\end{eqnarray}
From the above condition, we can conclude that $\ket{\phi^l_{i_1i_2\ldots i_N}}=0$ whenever $l_1+l_k+i_1+i_k$ mod $2=1$ for any $k=2,3,\ldots,N$. Thus, the state \eqref{genstate4} that satisfies the condition \eqref{ststate2} is given by
\begin{eqnarray}\label{genstate5}
   \ket{\tilde{\psi}_l}=\ket{l_1\ldots l_N}\ket{\phi_{l_1\ldots l_N}}+\ket{\overline{l}_1\ldots \overline{l}_N}\ket{\phi_{\overline{l}_1\ldots \overline{l}_N}},
\end{eqnarray}
where $l_i=0,1$ for any $i=1,2,\ldots,N$ and $\overline{l}_i$ denotes the negation of the bit $l_i$. Now, putting this state \eqref{genstate5} to the condition \eqref{SOSrel10}, we obtain that
\begin{eqnarray}
   (-1)^{l_1}\left[\ket{\overline{l}_1\ldots \overline{l}_N}\ket{\phi_{l_1\ldots l_N}}+ \ket{l_1\ldots l_N}\ket{\phi_{\overline{l}_1\ldots \overline{l}_N}}\right]=\ket{l_1\ldots l_N}\ket{\phi_{l_1\ldots l_N}}+\ket{\overline{l}_1\ldots \overline{l}_N}\ket{\phi_{\overline{l}_1\ldots \overline{l}_N}}.\nonumber\\
\end{eqnarray}
This implies that
\begin{eqnarray}
 \ket{\phi_{\overline{l}_1\ldots \overline{l}_N}}= (-1)^{l_1}\ket{\phi_{l_1\ldots l_N}},
\end{eqnarray}
which after substituting into Eq. \eqref{genstate5} allows us to conclude that 
the state $\ket{\tilde{\psi}_l}$ is given by
\begin{eqnarray}
   \ket{\tilde{\psi}_l}=\frac{1}{\sqrt{2}}\left(\ket{l_1l_2\ldots l_N}+(-1)^{l_1}\ket{\overline{l}_1\ldots \overline{l}_N}\right)_{A'}\otimes \ket{\phi_{l}}_{A''G},
\end{eqnarray}
where we have also included the normalization factor.
Tracing the $G$ subsystem out, we finally obtain 
\begin{eqnarray}\label{statecertified1}
\left(\bigotimes_{i=1}^N U_{A_i} \right)\rho^l_{{A}}\left(\bigotimes_{i=1}^N U_{A_i}^{\dagger} \right)=\proj{\phi_l}_{A'}\otimes\tilde{\rho}^l_{A''}.
\end{eqnarray}
where $\tilde{\rho}^l_{A''}$ denotes the auxiliary state acting on $\bigotimes_{i=1}^N\mathcal{H}_{A_i''}$. 

Putting it back into Eq. \eqref{54} and taking the unitaries to the right-hand side, we see that
\begin{eqnarray}\label{cond1}
\proj{\phi_l}_{A'}\otimes\tilde{\rho}^l_{A''}=2^N\Tr_{E}\left[\left(\I_{A}\otimes R_l\right)\bigotimes_{i=1}^N\proj{\tilde{\psi}_{A_iE_i}}\right],
\end{eqnarray}
where $\ket{\tilde{\psi}_{A_iE_i}}= U_{A_i}\ket{\psi_{A_iE_i}}$ and we used the fact that $\overline{P}(l|0)=1/2^N$.
\begin{center}
    \textbf{B.\ \ States generated by the source $\ket{\psi_{A_iE_i}}$}
\end{center}

Our aim now is to characterize the states produced by the sources $\ket{\psi_{A_iE_i}}$ and to this aim we follow the approach of Ref. \cite{sarkarPRL}. First, with the aid of the Schmidt decomposition each state $\ket{\tilde{\psi}_{A_iE_i}}= U_{A_i}\ket{\psi_{A_iE_i}}$ can be expressed in the following way
\begin{eqnarray}\label{Schmidt}
\ket{\tilde{\psi}_{A_iE_i}}=\sum_{j=0}^{d_{i}-1}\alpha_j\ket{e_j}_{A_i}\ket{f_j}_{E_i}
\end{eqnarray}
such that $d_{i}$ denotes the dimension of the Hilbert space $\mathcal{H}_{A_i}$ and $\{\ket{e_j}_{A_i}\}$ and $\{\ket{f_j}_{E_i}\}$ denotes the local orthonormal bases of the subsystem $A_i$ and $E_i$, respectively. Notice that, as already proven, the dimension $d_{i}$ is even for any $i$ as the Hilbert space of any external party decomposes as $\mathcal{H}_{A_i}=\mathbbm{C}^2\otimes\mathcal{H}_{A_i''}$. Moreover, the Schmidt coefficients satisfy $\alpha_j> 0$ and $\sum_j\alpha_j^2=1$.

Let us now consider a unitary operation $U_{E_i}:\mathcal{H}_{E_i}\rightarrow\mathcal{H}_{E_i}$ such that $U_{E_i}\ket{f_j}=\ket{e_j^*}_{E_i}$. After applying it to the state \eqref{Schmidt}, one arrives at
\begin{eqnarray}\label{genstate2}
\ket{\overline{\psi}_{A_iE_i}}= U_{E_i}\ket{\tilde{\psi}_{A_iE_i}}=\sum_{j=0}^{d_{i}-1}\alpha_j\ket{e_j}_{A_i}\ket{e_j^*}_{E_i}
\end{eqnarray}
Let us now define a matrix $P_{E_i}$ of rank $d_{i}$ as,
\begin{eqnarray}\label{P}
P_{E_i}=\sum_{j=0}^{d_{i}-1}\alpha_j\proj{e_j^*},
\end{eqnarray}
using which we rewrite the state \eqref{genstate2} as
\begin{eqnarray}\label{genstate3}
\ket{\overline{\psi}_{A_iE_i}}=\sqrt{d_i}\ \I_{A_i}\otimes P_{E_i}\ket{\phi^{+,d_{i}}_{A_iE_i}}\qquad \forall i.
\end{eqnarray}
Here $\ket{\phi^{+,d_{i}}_{A_iE_i}}$ denotes the maximally entangled state of local dimension $d_{i}$, that is,
\begin{eqnarray}
\ket{\phi^{+,d_{i}}_{A_iE_i}}=\frac{1}{\sqrt{d_{i}}}\sum_{j=0}^{d_{i}-1}\ket{e_j}_{A_i}\ket{e_j^*}_{E_i}=\frac{1}{\sqrt{d_{i}}}\sum_{j=0}^{d_{i}-1}\ket{j}_{A_i}\ket{j}_{E_i}.
\end{eqnarray}
Putting it back into the condition \eqref{cond1}, we get that
\begin{eqnarray}\label{cond2}
\proj{\phi_l}_{A'}\otimes\tilde{\rho}^l_{A''}=2^N\Tr_{E}\left[\left(\I_{A}\otimes \overline{R}_{l,E}\right)\bigotimes_{i=1}^N\proj{\phi^{+,d_{i}}_{A_iE_i}}\right]
\end{eqnarray}
where 
\begin{eqnarray}\label{rl}
\overline{R}_{l,E}=D_N\left(\bigotimes_{i=1}^NP_{E_i}U_{E_i}\right)\ R_l\ \left(\bigotimes_{i=1}^N U_{E_i}^{\dagger}P_{E_i}\right),
\end{eqnarray}
such that $D_N=\Pi_{i=1}^N d_{i}$.
Now, notice that the state $\bigotimes_{i=1}^N\ket{\phi^{+,d_{i}}_{A_iE_i}}$ can also be expressed as
a singe maximally entangled state between the external parties $A$ and $E$ of the local dimension $D_N$,
\begin{eqnarray}\label{MaxEntDN}
\bigotimes_{i=1}^N\ket{\phi^{+,d_{i}}_{A_iE_i}}=\ket{\phi^{+,D_N}_{AE}},
\end{eqnarray}
where
\begin{eqnarray}
\ket{\phi^{+,D_N}_{A|E}}=\frac{1}{\sqrt{D_N}}\sum_{j=0}^{D_N-1}\ket{j}_{A}\ket{j}_{E}.
\end{eqnarray}
Plugging this state back to the condition \eqref{cond2} and then using the well-known property of the maximally entangled state that $\I_A\otimes Q_B\ket{\phi^{+,D}_{AB}}=Q^T_A\otimes\I_B\ket{\phi^{+,D}_{AB}}$ for any matrix $Q$, we have that
\begin{eqnarray}
\proj{\phi_l}_{A'}\otimes\tilde{\rho}^l_{A''}=2^N\Tr_{E}\left[\overline{R}^T_{l,A}\otimes\I_{E}\proj{\phi^{+,D_N}_{A|E}}\right].
\end{eqnarray}
After taking the trace over $E$, the above formula simplifies to 
\begin{eqnarray}\label{cond3}
\proj{\phi_l}_{A'}\otimes\tilde{\rho}^l_{A''}=\frac{2^N}{D_N}\overline{R}^T_{l,A}.
\end{eqnarray}
We can now apply the transposition to both sides of the above relations and then expand its right-hand side using \eqref{rl}, to obtain 
\begin{eqnarray}\label{cond4}
\proj{\phi_l}_{A'}\otimes\tilde{\rho}^{l,T}_{A''}=2^N\left(\bigotimes_{i=1}^NP_{A_i}U_{{A_i}}\right)\ R_l\ \left(\bigotimes_{i=1}^NU_{{A_i}}^{\dagger} P_{{A}_i}\right),
\end{eqnarray}
which after taking the sum over $l$ and using the fact that the measurement operators $R_l$ sum up to the identity, leads further to
\begin{eqnarray}\label{Raimat}
\sum_{l=0}^{2^N-1}\proj{\phi_l}_{A'}\otimes\tilde{\rho}^{l,T}_{A''}=2^N\bigotimes_{i=1}^NP_{A_i}^2.
\end{eqnarray}
Now, taking a partial trace over the subsystems $A_2,A_3,\ldots,A_N$ in the above expression, we obtain
\begin{eqnarray}
\frac{\I_{A_1'}}{2}\otimes\sigma_{A_1''}=P_{A_1}^2,
\end{eqnarray}
which by virtue of the fact that $P_{A_i}\geq 0$ for any $i$, implies that 
\begin{eqnarray}
 \I_{A_1'}\otimes\sqrt{\frac{\sigma_{A_1''}}{2}}=P_{A_1}
\end{eqnarray}
where,
\begin{eqnarray}
 \sigma_{A_1''}=\frac{1}{2^N}\sum_{l=0}^{2^N-1}\Tr_{A_2''\ldots A_N''}\left(\tilde{\rho}^{l,T}_{A''}\right).
\end{eqnarray}
In a similar way we can determine from Eq. (\ref{Raimat}) the form of the other matrices $P_{A_{i}}$
and thus 
\begin{eqnarray}\label{Pj}
P_{A_i}=\I_{A_i'}\otimes\sqrt{\frac{\sigma_{A_i''}}{2}}\qquad (i=1,\ldots,N),
\end{eqnarray}
where,
\begin{eqnarray}\label{sigma1}
 \sigma_{A_j''}=\frac{1}{2^N}\sum_{l=0}^{2^N-1}\Tr_{A''\setminus\{A_{j}''\}}\left(\tilde{\rho}^{l,T}_{A''}\right).
\end{eqnarray}
Plugging $P_{A_j}$ in the state \eqref{genstate3} and also recalling that $d_{i}=2d'_i$, we obtain 
\begin{eqnarray}\label{66}
 \ket{\overline{\psi}_{A_iE_i}}=\sqrt{d'_{i}}\ \left(\I_{A_iE_i'}\otimes\sqrt{\sigma_{E_i''}}\right) \ket{\phi^{+,d_{i}}_{A_iE_i}}\qquad \forall i.
\end{eqnarray}
Using then again the fact that $d_i$ is even, which implies that the maximally entangled state
$\ket{\phi^{+,d_{i}}_{A_iE_i}}$ can also be written as
\begin{equation}
\ket{\phi^{+,d_{i}}_{A_iE_i}}=\ket{\phi^{+}_{A_i'E_i'}}\ket{\phi^{+,d'_{i}}_{A_i''E_i''}}, 
\end{equation}
we finally conclude that there are unitary operations $U_{A_i}$ and $U_{E_i}$ such that
\begin{eqnarray}
  (U_{A_i}\otimes U_{E_i})\ket{\psi_{A_iE_i}}=\ket{\overline{\psi}_{A_iE_i}}=\ket{\phi^{+}_{A_i'E_i'}}\ket{\xi_{A_i''E_i''}}\qquad \forall i.
\end{eqnarray}
where the auxiliary state $\ket{\xi_{A_i''E_i''}}$ is given by,
\begin{eqnarray}\label{junkst1}
  \ket{\xi_{A_i''E_i''}}= \left(\I_{A_i''}\otimes\sqrt{d'_i\sigma_{E_i''}}\right) \ket{\phi^{+,d'_{i}}_{A''_iE''_i}}.
\end{eqnarray}

\begin{center}
    \textbf{C.\ \ Entangled measurement  $\{R_l\}$}
\end{center}

Let us now concentrate on the measurement $\{R_l\}$. 
To this end, we exploit the explicit form of $P_{A_i}$ and fact that the matrices $\sigma_{A_{i}''}$ are invertible (which is a consequence of the assumption that the states produced by the sources are locally full rank) to rewrite Eq. (\ref{cond4}) as
%
%
\iffalse
For this purpose, let us first notice from Eq. 
\eqref{sigma1} that
%
\begin{eqnarray}
  \sigma_{A_j''}\leqslant  \frac{1}{2^N}\sum_{l=0}^{2^N-1}\I_{A_j''}=\I_{A_j''}\qquad \forall j,
\end{eqnarray}
%
where we used the fact that $\Tr_{A''\setminus\{A_{j}''\}}\left(\tilde{\rho}^l_{A''}\right)\leqslant \I_{A_j''}$. Also notice that the states $\sigma_{A_j''}$ are invertible as the states $\ket{\psi_{A_j\overline{A}_j}}$ are full-rank. 
%
Now, revisiting the relation \eqref{cond4} and putting in $P_{A_i}$ from \eqref{Pj}, we observe that
%
%
\fi
%
\begin{eqnarray}
\left(\bigotimes_{i=1}^NU_{{A_i}}\right)\ R_l\ \left(\bigotimes_{i=1}^N U_{{A_i}}^{\dagger}\right)=\proj{\phi_l}_{A'}\otimes\left(\bigotimes_{i=1}^N\sigma^{-1/2}_{A_i''}\right)\tilde{\rho}^{l,T}_{A''}\left(\bigotimes_{i=1}^N\sigma^{-1/2}_{A_i''}\right),
\end{eqnarray}
and further in a simpler way as
\begin{eqnarray}\label{Rl}
 \left(\bigotimes_{i=1}^N U_{{A_i}}\right)\ R_l\  \ \left(\bigotimes_{i=1}^NU_{{A_i}}^{\dagger}\right)=\proj{\phi_l}_{A'}\otimes(\tilde{R}_{l})_{A''},
\end{eqnarray}
where the matrices $\tilde{R}_{l}$ acting on the $A_{i}''$ subsystems are defined as
\begin{eqnarray}\label{Rl1} \tilde{R}_{l,A''}=\left(\bigotimes_{i=1}^N\sigma^{-1/2}_{A_i''}\right)\tilde{\rho}^{l,T}_{A''}\left(\bigotimes_{i=1}^N\sigma^{-1/2}_{A_i''}\right).
\end{eqnarray}
%
%Notice that every element $R_l\leqslant \I$ due to which we can conclude %from the above condition \eqref{Rl} that $ \tilde{R}_l\leqslant \I$. 

Now, an important step is to observe that after implementing the fact that $\sum_lR_l=\I$, Eq. (\ref{Rl}) allows us to conclude that
\begin{eqnarray}
  \I_{A}=\sum_{l=0}^{2^N-1}\proj{\phi_l}_{A'}\otimes(\tilde{R}_{l})_{A''}.
\end{eqnarray}
Then, since the vectors $\ket{\phi_l}$ are mutually orthogonal, 
the above directly implies that 
%
%Now, again using the fact that $\sum_{l=0}^{2^N-%1}\proj{\phi_l}_{A'}=\I_{A'}$, we can rearrange the terms in the above %condition to get
%\begin{eqnarray}
%  \sum_{l=0}^{2^N-1}\proj{\phi_l}_{A'}\otimes\left(\I_{A''}-%\tilde{R}_{l,A''}\right)=0.
%\end{eqnarray}
%As $ \tilde{R}_l\leqslant \I$ for any $l$ and all the vectors %$\ket{\phi_l}$ are mutually orthogonal, the only possible solution of the %above condition is that
%
\begin{eqnarray}\label{Rl2}
  \tilde{R}_{l}=\I_{A''}
\end{eqnarray}
using which we finally obtain from Eq. \eqref{Rl} that the measurement operators $R_l$ are up to a local unitary transformation given by
\begin{eqnarray}
  \left(\bigotimes_{i=1}^N U_{{A_i}}\right)\ R_l\  \ \left(\bigotimes_{i=1}^NU_{{A_i}}^{\dagger}\right)=\proj{\phi_l}_{A'}\otimes\I_{A''}\qquad\forall l.
\end{eqnarray}
%
%Notice here that we prove that the supports of all the states are the same.
\\

\begin{center}
    \textbf{D.\ \ Measurements $A_{i,2}$}
\end{center}

Let us finally provide a self-testing statement for the third measurement of the external parties $A_{i,2}$ for any $i$. For this purpose, we consider Eq. \eqref{SOSrel2} and substitute into it the certified observables and states from Eqs. \eqref{mea1} and \eqref{statecertified1} to obtain a set of equations
\begin{equation}\label{SOSrel22}
           \left[ -(-1)^{l_i}\overline{A}_{1,2}\overline{A}_{i,2}\otimes \bigotimes_{\substack{j=2\\j\ne i}}^{N} X_{A_j'}\right]\proj{\phi_l}_{A'}\otimes\tilde{\rho}^l_{A''}=\proj{\phi_l}_{A'}\otimes\tilde{\rho}^l_{A''}\qquad i=2,\ldots,N.
\end{equation}
where $\overline{A}_{i,2}=U_{A_i}A_{i,2}\,U_{A_i}^{\dagger}$. As the local Hilbert spaces are even-dimensional, we can express the observables $\overline{A}_{i,2}$ as
\begin{eqnarray}\label{genmea2}
    \overline{A}_{i,2}=\I\otimes P_{i}+Z\otimes Q_{i}+X\otimes R_{i}+Y\otimes S_{i},
\end{eqnarray}
where $P_{i}, Q_{i}, R_{i}, S_{i}$ are hermitian matrices that mutually commute and satisfy $P_{i}^2+Q_{i}^2+R_{i}^2+S_{i}^2=\I$ for any $i$.

Let us first look at the case when $N=2$. Putting the form of the measurements \eqref{genmea2} into the relation \eqref{SOSrel22} and then sandwiching it with $\ket{\phi_l}$, we obtain  
%the following two conditions
\begin{eqnarray}\label{neweq31}
     \left[P_{1}\otimes P_{2}+(-1)^{l_1+l_2}Q_{1}\otimes Q_{2}+(-1)^{l_1}R_{1}\otimes R_{2} -(-1)^{l_2}S_{1}\otimes S_{2}\right]\tilde{\rho}^l_{A_1''A_2''}= -(-1)^{l_2}\tilde{\rho}^l_{A_1''A_2''}
\end{eqnarray} 
for any $l_1,l_2=0,1$. Let us now observe that from 
Eq. (\ref{Rl1}) and the fact that $\tilde{R}_l=\mathbbm{1}_{A''}$
it follows that the matrices $\tilde{\rho}^l_{A_1''A_2''}$ are full rank
and therefore the above implies the following four matrix equations 
\begin{eqnarray}\label{neweq3}
     P_{1}\otimes P_{2}+(-1)^{l_1+l_2}Q_{1}\otimes Q_{2}+(-1)^{l_1}R_{1}\otimes R_{2} -(-1)^{l_2}S_{1}\otimes S_{2}= -(-1)^{l_2}\mathbbm{1},
\end{eqnarray}
we can directly be solved to obtain
\begin{equation}\label{neweq4}
    S_{1}\otimes S_{2}=\I,\qquad  P_{1}\otimes P_{2}= Q_{1}\otimes Q_{2}= 
    R_{1}\otimes R_{2}=0.
\end{equation}

Now, considering the general case $N\geq 3$ and again putting form of the measurements \eqref{neweq3} into \eqref{SOSrel22}, 
%we obtain
%\begin{eqnarray}\label{neweq1}
%\ket{\phi_{l}}\!\!\bra{\phi_l}\otimes\left[P_{1}\otimes P_{i}+(-1)^{l_1+l_i}Q_{1}\otimes Q_{i}\right]\tilde{\rho}^l_{A''}+   \proj{\phi_{l}}\otimes\left[(-1)^{l_1}R_{1}\otimes R_{i} -(-1)^{l_i}S_{1}\otimes S_{i}\right]\tilde{\rho}^l_{A''}+\ldots\nonumber\\= -(-1)^{l_i}\proj{\phi_{l}}\otimes\tilde{\rho}^l_{A''}\qquad i=2,\ldots,N.
%\end{eqnarray}
%where $\ket{\phi_{l}}=\I_{1}\otimes\I_{A_i}\otimes\bigotimes_{\substack{j=2\\j\ne i}}^{N} X_{A_j'}\ket{\phi_l}$ and the dots represent the rest of the terms. Now,
and sandwiching the above expression %\eqref{neweq1} with 
with the state $\ket{\phi_l}$, we obtain for any $l_i=0,1$ 
%the following two conditions
\begin{eqnarray}
     \left[(-1)^{l_1}R_{1}\otimes R_{i} -(-1)^{l_i}S_{1}\otimes S_{i}\right]\tilde{\rho}^l_{A''}= -(-1)^{l_i}\tilde{\rho}^l_{A''}\qquad (i=2,\ldots,N)
\end{eqnarray}
%and,
%\begin{eqnarray}
 %    \left[P_{1}\otimes P_{i}+(-1)^{l_1+l_i}Q_{1}\otimes Q_{i}\right]\tilde{\rho}^l_{A''}= 0\qquad i=2,\ldots,N.
%\end{eqnarray}
As $\tilde{\rho}^l_{A''}$ is full-rank, the above translates to 
\begin{eqnarray}\label{neweq2}
    (-1)^{l_1}R_{1}\otimes R_{i} -(-1)^{l_i}S_{1}\otimes S_{i}= -(-1)^{l_i}\I.
\end{eqnarray}
For any $i$, we get two different equations corresponding to two values of $l_1=0$ and $l_1=1$, which directly imply
\begin{equation}\label{neweq5}
    S_{1}\otimes S_{i}=\I,\qquad  %P_{1}\otimes P_{i}= Q_{1}\otimes Q_{i}= 
    R_{1}\otimes R_{i}=0 \qquad (i=2,\ldots,N).
\end{equation}
Now, one notices that $S_{1}^2\otimes S_{i}^2=\I$ which together with the upper bound $S_i^2\leq \mathbbm{1}$ which follows from the condition $P_{i}^2+Q_{i}^2+R_{i}^2+S_{i}^2=\I$, allows us to conclude that 
$S_i^2=\mathbbm{1}$ for any $i$. This also means that $P^2_i=R^2_i=Q^2_i=0$ and consequently that 
\be
P_i=R_i=Q_i=0.
\ee
This analysis leads us to the conclusion that 
$ \tilde{A}_{i,2}=Y\otimes S_i$ with $S_i^2=\mathbbm{1}$. It turns out, however, that the conditions \eqref{neweq5} are more restrictive. To see that explicitly, let us follow the reasoning applied already in Ref. \cite{sarkarPRL} and observe that each $S_i$ can be decomposed us
$S_i=S_i^{+}-S_i^{-}$, where $S_i^{\pm}$ are projections onto the 
eigensubspaces of $S_i$ corresponding to the eigenvalues $\pm1$ with $S_i^{+}+S_i^{-}=\I$. After plugging this representation into $S_1\otimes S_i=\mathbbm{1}$ 
and rearranging terms one arrives at the following conditions
\begin{equation}
    S_1^{+}\otimes S_i^-=S_1^{-}\otimes S_i^+=0 \qquad (i=2,\ldots,N).
\end{equation}
These conditions imply that there are two possible solutions 
for Alice's third measurements: either
\begin{equation}
    \tilde{A}_{i,2}=Y_{A_i'}\otimes \I_{A_i''} \qquad (i=1,\ldots,N) 
\end{equation}
or
\begin{equation}
    \tilde{A}_{i,2}=-Y_{A_i'}\otimes \I_{A_i''} \qquad (i=1,\ldots,N).
\end{equation}
This completes the proof.
\end{proof}

\section{Appendix B: Self-testing any extremal POVM}
\label{AppC}
%

%%%%%%%%%%%%%%%%%%%%%%%%%%%%%%%%%%%%%%%%%%%%%%%%%%%%%%%%%%%%%%%%%%%%%%%%%%%%%%%%%%%%%%%%
%%%%%%%%%%%%%%%%%%%%%%%%%%%%%%%%%%%%%%%%%%%%%%%%%%%%%%%%%%%%%%%%%%%%%%%%%%%%%%%%%%%%%%%%%%
%%%%%%%%%%%%%%%%%%%%%%%%%%%%%%%%%%%%%%%%%%%%%%%%%%%%%%%%%%%%%%%%%%%%%%%%%%%%%%%%%%%%%%%%%

%\subsection{Any extremal POVM}

Let us then consider an ideal reference measurement $E'_1=\{R'_{l|1}\}$ with $R'_{l|1}$ $(l=0,\ldots,K-1)$ being measurement elements defined on some Hilbert space $\mathbbm{C}^D$ and assume this measurement to be extremal. For completeness let us recall that an extremal POVM is one that cannot be decomposed into a convex combination of other POVM's. As $R'_{l|1}$ are of arbitrary rank, let denote them by $r_l=\mathrm{rank}(R'_{l|1})$. Let us also assume
that the measurement has $K\leq 2^N$ outcomes; $r_1+\ldots+r_K=D$. Such a measurement can always be embedded in an $N$-qubit Hilbert space 
$(\mathbbm{C}^2)^{\otimes N}$ where $N$ the minimal 
natural number such that $D\leq 2^N$. In order to represent $R'_{l|1}$ in the $N$-qubit Hilbert space one can for instance assign to any element of the
standard basis $\ket{i}\in \mathbbm{C}^D$ 
an element $\ket{i_1\ldots i_N}$ of the product standard basis of $(\mathbbm{C}^2)^{\otimes N}$
such that $i_1\ldots i_N$ is a binary representation of $i$ $(i=0,\ldots,D-1)$. 
Then, one needs to add another projective measurement element $M^{\perp}$ to $E_1'$ which is orthogonal to the support of $E_1'$ so that the resulting measurement defined on $(\mathbbm{C}^2)^{\otimes N}$ satisfies
\begin{equation}
    \sum_{i=1}^KR'_{i|1}+M^{\perp}=\mathbbm{1}_{2^N},
\end{equation}
where $\mathbbm{1}_{2^N}$ is the identity acting on $(\mathbbm{C}^2)^{\otimes N}$ where $\mathrm{rank}(M^{\perp})=2^N-D$.
Let us denote the new $(K+1)-$outcome measurement  by 
$N=\{R'_{1|1},\ldots,R'_{K|1},M^{\perp}\}$ and notice that 
it is extremal too because it is obtained by complementing the extremal measurement $E_1'$ by a projector orthogonal to its support \cite{APP05}.

Let us now represent any measurement element $N_l$ in terms of the Pauli basis
\begin{eqnarray}\label{idealPauli}
\sigma_{0}=Z,\quad\sigma_{1}=X,\quad\sigma_{2}=Y,
\quad \sigma_{3}=\I_2,
\end{eqnarray}
as
\begin{eqnarray}\label{proj1}
N_l=\sum_{i_1,\ldots,i_N=0}^{3}f^l_{i_1,\ldots,i_N}\bigotimes_{k=1}^N\sigma_{k,i_k}.
\end{eqnarray}

%As $\proj{\delta_{b,l}}$ is Hermitian, $f_{i_1,\ldots,i_N}(\vec{\alpha}_{b,l})$ is real for any $i_1,\ldots,i_N$.

Let us again consider the quantum network scenario presented and assume that the observed correlations additionally satisfy the following conditions
\begin{eqnarray}\label{betstatfull}
\left\langle\tilde{A}_{1,i_1}\otimes\bigotimes_{k=2}^N A_{k,i_k}\otimes (R_{l|1})_E\right\rangle_{\psi_{AE}}&=&f^l_{i_1,\ldots,i_N}\quad \forall l,i_1,\ldots,i_N,
\end{eqnarray}
where $\tilde{A}_{1,0}$ and $\tilde{A}_{1,1}$ are given in \eqref{overA} with $\tilde{A}_{1,2}=A_{1,2}$, $\tilde{A}_{1,3}=\I$ and $A_{k,3}=\I$ for any $k$. 

%, or
%\begin{eqnarray}\label{betstatfull1}
%\left\langle\tilde{A}_{1,i_1}\bigotimes_{k=2}^N A_{k,i_k}R_{l|1}\right\rangle&=&g(i_1,\ldots,i_N)\frac{f_{l,i_1,\ldots,i_N}}{2^N}\quad \forall l,i_1,\ldots,i_N
%\end{eqnarray}
%where $g(i_1,\ldots,i_N)=1$ for even number of $j$ and $g(i_1,\ldots,i_N)=-1$ for odd number of $j$ such that $i_j=2$.

%
Notice that the above statistics can be realised if the sources generate the two-qubit maximally entangled state $\ket{\phi^+}=(1/\sqrt{2})(\ket{00}+\ket{11})$ and the measurement $E_1=\{R_{l|1}\}$ is exactly $R_{l|1}=(R'_{l|1})^*$ where $R'_{l|1}$ are the reference measurement operators and $*$ denotes the complex conjugate, whereas the external parties perform the measurements in Eq. \eqref{GHZObs}.

Next, we show that the above conditions along with certification of the states and measurements presented in Theorem \ref{theorem1} are sufficient to show that up to the standard equivalences and complex conjugation the unknown measurement $E_{1}=\{R_{l|1}\}$ is equivalent to the ideal one $E_1'\equiv N$. %The theorem stated below is labeled as Theorem 2 in the manuscript.

\begin{thm}\label{theorem3}
Let us suppose that the correlations generated in the quantum network satisfy the assumptions of Theorem \ref{theorem1} along with the additional constraints in Eq. (\ref{betstatfull}). Then, for any $l=0,\ldots,K-1$ we can conclude that 
\begin{eqnarray}
    \label{stmeaNLWE2}
U_E R_{l|1} U_E^{\dagger} &=&(R'_{l|1})^* \otimes \mathbbm{1}_{E''},\quad \mathrm{if}\quad A_{i,2}=Y\otimes\I \quad\forall i, \quad\text{or}\quad \nonumber\\
 U_E R_{l|1} U_E^{\dagger} &=& R'_{l|1} \otimes \mathbbm{1}_{E''}\quad \ \ \ \mathrm{if}\quad  A_{i,2}=-Y\otimes\I \quad\forall i
\end{eqnarray}
where $U_E$ is the same unitary as in Theorem \ref{theorem1}
and $E''=E_1''\ldots E_N''$.
\end{thm}

\begin{proof}
For simplicity, we represent $R_{l|1}\equiv R_{l}$ throughout the proof. 
Let us consider the relation in Eq. (\ref{betstatfull}) and then substitute into it the observables $A_{s,j}$ for $s=1,\ldots,N$ and $j=0,1,2$ given in Eqs. \eqref{stmea1} and \eqref{stmea2}. As $A_{s,2}$ for any $s$ can have two realisations as given in Eq. \eqref{stmea2}, we first consider the case when $A_{s,2}=Y\otimes\I_{s''}$. This gives us
\begin{eqnarray}\label{stcond1NLWE}
\left\langle\left(\bigotimes_{s=1}^{N} U_{A_s}^{\dagger}\right) \left(\bigotimes_{k=1}^N\sigma_{k,i_k}\otimes\I_{A''}\right)\left(\bigotimes_{s=1}^{N} U_{A_s}\right)\otimes R_l\right\rangle_{\psi_{AE}}=f^l_{i_1,\ldots,i_N}.
\end{eqnarray}
It then follows from Theorem \ref{theorem1} that the global state $\ket{\psi_{AE}}$ can be represented as in Eq. \eqref{A10}, which allows us to rewrite the above condition in the following form
%
%
%\begin{eqnarray}\label{state1}
%\ket{\psi_{AE}}=\bigotimes_{s=1}^{N}\ket{\psi_{A_s\%overline{A}_s}}=\bigotimes_{s=1}^{N} %U_{A_s}^{\dagger}\otimes %U_{\overline{A}_s}^{\dagger}|\phi^+_{A_s'\overline{%A}_s'}\rangle\otimes\ket{\xi_{A_s''\overline{A}_s''%}}.
%\end{eqnarray}
%
%Notice also that by virtue of Eq. \eqref{junkst1} %the junk states %$\ket{\xi_{A_s''\overline{A}_s''}}$ can be %represented as
%
%\begin{eqnarray}\label{junkst2}
% \ket{\xi_{A_s''\overline{A}_s''}}= 
%(\I_{A_s''}\otimes P_{\overline{A}_s''} %)|\phi^{+}_{d_{s}''}\rangle_{A_s''\overline{A}_s''}%,
%\end{eqnarray}
%
%where %$P_{\overline{A}_s''}=\sqrt{d_s''\sigma_{\overline{%A}_s''}}$.
%The joint state in Eq. \eqref{state1} can be %further written as [cf. Eq. \eqref{junkst2}],
%
%\begin{eqnarray}\label{115}
%\left(\bigotimes_{s=1}^{N}U_{A_s}\otimes %U_{\overline{A}_s}\right)|\psi_{A|E}\rangle=\left(\%bigotimes_{s=1}^{N} %P_{\overline{A}_s''}\right)|\phi^{+}_{2^ND}\rangle_%{A|E}.
%\end{eqnarray}
%
%Here the local dimension of the state is $2^ND$ %where $D=\prod_sd_s''$ and $A|E$ denotes the %bipartition between the subsystem $A$ and %$E\equiv\overline{A_1}\ldots\overline{A_N}=\overlin%e{A}$. Plugging the state \eqref{115} into Eq. %\eqref{stcond1NLWE}, and then using the fact that %$(\I\otimes Q)|\phi^{+}_{d}\rangle=
%(Q^T\otimes\I)|\phi^{+}_{d}\rangle$ for any matrix %$Q$, we obtain
%
\begin{eqnarray}\label{116}
d_1'\ldots d_N'\left\langle \left(\bigotimes_{k=1}^N\sigma_{k,i_k}\otimes\bigotimes_{s=1}^{N}\sigma_{A_{s}''}^T\right)\overline{R}_l^T\otimes\I_E\right\rangle_{\phi^{+}_{2^ND}}=f^l_{i_1,\ldots,i_N},
\end{eqnarray}
where we have also used the fact that 
$(\I\otimes Q)|\phi^{+}_{d}\rangle=
(Q^T\otimes\I)|\phi^{+}_{d}\rangle$ and 
$\overline{R}_l$ are given by
\begin{eqnarray}\label{R0eq}
\overline{R}_l=\left(\bigotimes_{s=1}^{N} U_{\overline{A}_s} \right) R_l\left(\bigotimes_{s=1}^{N} U_{\overline{A}_s}^{\dagger}\right).
\end{eqnarray}
After tracing Eve's subsystems the expression in Eq. \eqref{116} can further be simplified to 
\begin{eqnarray}
\Tr\left[ \left(\bigotimes_{k=1}^N\sigma_{k,i_k}\otimes\bigotimes_{s=1}^{N} \sigma_{A_{s}''}^T\right)\overline{R}_l^T\right]=2^Nf^l_{i_1,\ldots,i_N}.
\end{eqnarray}

Now, we can exploit the fact that $\overline{R}_l^T$ acts on the Hilbert space $\mathbbm{C}^{2^N}\otimes\bigotimes_{s=1}^{N}\mathcal{H}_{{A_s''}}$, we can again express it using the Pauli basis as
\begin{eqnarray}\label{genR}
\overline{R}_l^T=\sum_{i_1,\ldots,i_N=0}^{3}
\bigotimes_{k=1}^N\sigma_{k,i_k}\otimes \tilde{R}^l_{i_1,\ldots,i_N}
\end{eqnarray}
where $\tilde{R}_{i_1,\ldots,i_N}$ acts on $\bigotimes_{s=1}^{N}\mathcal{H}_{{A_s''}}$.
Then, using the fact that Pauli matrices are orthogonal in the Hilbert-Schmidt scalar product,
$\Tr[\sigma_i\sigma_j]=2\delta_{ij}$ with $\delta_{ij}$ being the Kronecker's symbol, we obtain 
\begin{eqnarray}\label{nlwest11}
      %2^N
      \Tr\left[ \left(\bigotimes_{s=1}^{N}\sigma^T_{A_s''}\right)\tilde{R}^l_{i_1,\ldots,i_N}\right]=%2^N
      f^l_{i_1,\ldots,i_N}\qquad \forall i_1,\ldots,i_N.
\end{eqnarray}
The rest of the proof follows the same lines as the proof of Theorem 6 of Ref. \cite{sarkar3}. We first decompose the state $\bigotimes_{s=1}^{N}\sigma^T_{A_s''}$ in its eigenbasis, let's say $\{\ket{b}\}$ with eigenvalues $p_b$ which allows us to obtain from Eq. \eqref{nlwest11} that
\begin{eqnarray}\label{POVMST5}
\sum_bp_b\bra{b}\tilde{R}^l_{i_1,\ldots,i_N}\ket{b}=f^l_{i_1,\ldots,i_N}.
\end{eqnarray}
Next, we introduce a collection of POVMs $\{R_{l}^b\}_l$ for any $b$, the measurement operators of which are given by 
\begin{eqnarray}
    R_{l}^b&=&\Tr_{A_1''\ldots A_N''}\left[ \left(\I_{A'}\otimes\ket{b}\!\bra{b}_{A_1''\ldots A_N''}\right)\overline{R}_l^T\right]\nonumber\\
&=&\sum_{\textbf{}i_1,\ldots,i_N=0}^{3}\bra{b}\tilde{R}^l_{i_1,\ldots,i_N}\ket{b}
\bigotimes_{k=1}^N\sigma_{k,i_k}.
\end{eqnarray}
Since $\overline{R}_l^T\geq 0$ it directly follows from the above equation that $R_{l}^b\geq 0$ for any $l$ and $b$. Moreover, $\sum_l\overline{R}_l^T=\I$ implies that $\sum_l R_{l}^b=\I$ for any $b$. Thus, $\{R_{l}^b\}_l$ are proper quantum measurements for all $b$. 

These additional POVMs, through Eq. (\ref{POVMST5}), allow us to decompose
the ideal POVM $\{N_l\}$ as $N_{l}=\sum_bp_bR_{l}^b$.
As the ideal POVM is extremal, it can not be decomposed in terms of other POVM's. Consequently, we have that 
\begin{eqnarray}
\forall b\qquad R_{l}^b=N_l,
\end{eqnarray}
which is equivalent to
\begin{eqnarray}\label{eq83}
\bra{b}\tilde{R}^l_{i_1,\ldots,i_N}\ket{b}=f^l_{i_1,\ldots,i_N}\qquad \forall b .
\end{eqnarray}
Let us now consider the following vectors:
\begin{eqnarray}\label{eq84}
\ket{\varphi_{a,s,t}}=\frac{1}{\sqrt{2}}\left(\ket{s}\pm\mathbbm{i}^a\ket{t}\right),
\end{eqnarray}
such that $\ket{s}$ and $\ket{t}$ are vectors that belong to the eigenbasis $\{\ket{k}\}$ of $\bigotimes_{s=1}^{N}\sigma^T_{A_s''}$ such that $s\ne t$ and $a=0,1$. Let us look at the quantity
\begin{eqnarray}
\Tr_{A_1''\ldots A_N''}\left[\left(\I_{A'}\otimes|\varphi_{a,s,t}\rangle\!\langle \varphi_{a,s,t}|_{A_1''\ldots A_N''}\right)\overline{R}_l^T\right]=\sum_{i_1,\ldots,i_N=0}^{3}
\bigotimes_{k=1}^N\sigma_{k,i_k}\ \Tr(|\varphi_{a,s,t}\rangle\!\langle \varphi_{a,s,t}|_{A_1''\ldots A_N''}\tilde{R}^l_{i_1,\ldots,i_N}),
\end{eqnarray}
which with the aid of the explicit form of the vector \eqref{eq84}
can be rewritten as
\begin{eqnarray}
\Tr_{A_1''\ldots A_N''}\left[\left(\I_{A'}\otimes|\varphi_{a,s,t}\rangle\!\langle \varphi_{a,s,t}|_{A_1''\ldots A_N''}\right)\overline{R}_l^T\right]=N_l\pm
\Tr_{A_1''\ldots A_N''}\left[(\mathbbm{1}_{A'}\otimes L^a_{A_1''\ldots A_N''})\overline{R}_l^T\right],
\end{eqnarray}
where $L_{A_1''\ldots A_N''}^a=(\mathbbm{i}^a/2)\left(|t\rangle\!\langle s|+(-1)^a|s\rangle\!\langle t|\right)$. Using the fact that $\overline{R}_l^T$ is positive semi-definite, which implies that the left-hand side of the above 
relation is positive semi-definite too, we conclude that
\begin{eqnarray}\label{eq87}
N_l\geq
\pm\Tr_{A_1''\ldots A_N''}\left[(\mathbbm{1}_{B'}\otimes L^a_{A_1''\ldots A_N''})\overline{R}_l^T\right]=:\pm\Omega_{l}^a.
\end{eqnarray}
As we show at the end of the proof, the above operator inequality can be satisfied only if 
$\mathrm{supp}(\Omega_{l}^a)\subseteq \mathrm{supp}(N_l)$ for any $l$ and for both values of $a=0,1$. Let us then denote by $\ket{\phi_{l,m}}$ $(m=1,\ldots,r_l)$ as the eigenvectors of $N_l$, where $r_l$ denotes the rank of $N_l$. Due to the fact that $\mathrm{supp}(\Omega_{l}^a)\subseteq \mathrm{supp}(N_l)$, each operator $\Omega_{l}^a$ can be written using the eigenvectors of $N_l$ as
\begin{equation}\label{Omega}
    \Omega_l^a=\sum_{m,n=1}^{r_l}\gamma_{m,n}^{(l,a)}|\phi_{l,m}\rangle\!\langle\phi_{l,n}|,
\end{equation}
where $\gamma_{m,n}^{(l)}$ are some complex coefficients; recall that $\Omega_l^a$ are Hermitian.

%Importantly, as proven in Ref. \cite{APP05}, the fact that the measurement $N$ is extremal %implies that the supports of the measurement operators $N_l$ are mutually disjoint, meaning that %for no pair of operators $N_i$ and $N_j$ such that $i\neq j$ there exist a vector belonging to %the supports of both of them. As a consequence, the supports of $\Omega_{l}^a$ are also mutually %disjoint. 

Let us now notice that it follows from Eq. (\ref{eq87}) that 
\begin{eqnarray}
    \sum_{l}\Omega_l^a&=&\sum_l\Tr_{A_1''\ldots A_N''}\left[(\mathbbm{1}_{A'}\otimes L^a_{A_1''\ldots A_N''})\overline{R}_l^T\right]\nonumber\\
    &=& \Tr_{A_1''\ldots A_N''}\left[(\mathbbm{1}_{A'}\otimes L^a_{A_1''\ldots A_N''})\right]=0,
\end{eqnarray}
where the second equality follows from the fact that $\overline{R}_l^T$ sum up to the identity, whereas the third equality is a consequence of the fact that the $L^a$ operators are traceless.
We thus have that for any $a$, 
\begin{equation}
    \sum_{l}\Omega_l^a=0.
\end{equation}
After plugging Eq. (\ref{Omega}), we obtain the following equation
\begin{equation}\label{equacioni}
    \sum_{l,m,n}\gamma_{m,n}^{(l,a)}|\phi_{l,m}\rangle\!\langle\phi_{l,n}|=0
\end{equation}
Let us now exploit the fact that the measurement $\{N_l\}$ is extremal. 
As proven in Ref. \cite{APP05} it follows from the latter that 
the operators $|\phi_{l,m}\rangle\!\langle\phi_{l,n}|$ are linearly independent for all $l$, $m$ and $n$. It thus follows from Eq. (\ref{equacioni}) 
that $\gamma$'s vanish and therefore
\begin{equation}
    \Omega_{l}^a=\Tr_{A_1''\ldots A_N''}\left[(\mathbbm{1}_{A'}\otimes L^a_{A_1''\ldots A_N''})\overline{R}_l^T\right]=0
\end{equation}
for any $a$ and any $l$. Using now the explicit form of the operators $L^a_{A_1''\ldots A_N''}$, the above 
implies the following conditions for the $\overline{R}_l^T$ operators,
\begin{eqnarray}\label{POVMST3}
\bigotimes_{k=1}^N\sigma_{k,i_k}\left(\bra{s}\tilde{R}^l_{i_1,\ldots,i_N}\ket{t}+\bra{t}\tilde{R}^l_{i_1,\ldots,i_N}\ket{s}\right)=0,
\end{eqnarray}
for $a=0$, and 
\begin{eqnarray}\label{POVMST4}
\bigotimes_{k=1}^N\sigma_{k,i_k}\left(\bra{t}\tilde{R}^l_{i_1,\ldots,i_N}\ket{s}
-\bra{s}\tilde{R}^l_{i_1,\ldots,i_N}\ket{t}\right)=0
\end{eqnarray}
for $a=1$, where we used the fact that $\bigotimes_{k=1}^N\sigma_{k,i_k}$ are linearly 
independent for all $i_1,\ldots,i_N$. The only possible solution to these equations
is that
\begin{equation}
 \bra{s}\tilde{R}^l_{i_1,\ldots,i_N}\ket{t}=0\qquad (s\neq t).   
\end{equation}
Combining this with Eq. \eqref{eq83} we obtain that 
\begin{equation}
    \tilde{R}^l_{i_1,\ldots,i_N}=f_{i_1,\ldots,i_N}^l\mathbbm{1},
\end{equation}
which implies the desired relation $\overline{R}_l^T=N_l\otimes\I_{A_1''\ldots A_N''}$ 
for all $l$. 

Let us finally prove that $\mathrm{supp}(\Omega_l^a)\subseteq\mathrm{supp}(N_l)$ for any $l$ and $a$. For this purpose, let us pick a particular $N_l$ and consider its eigenvectors 
$\ket{\psi_i}$ $(i=1,\ldots,r_l)$, where $r_l$ is the rank of $N_l$. We then construct an orthonormal basis $\{\ket{\phi_i}\}_i$ with $i=1,\ldots,2^N$ in the $N$-qubit Hilbert space such that $\ket{\phi_i}=\ket{\psi_i}$ for $i=1,\ldots,r_l$.

Sandwiching then inequalities \eqref{eq87} with $\ket{\phi_i}$ and $\ket{\phi_j}$, 
one obtains
\begin{eqnarray}\label{VinoTinto}
\forall_{i,j}\qquad \bra{\phi_i}N_l\ket{\phi_j}\geq
\pm\bra{\phi_i}\Omega_l^a\ket{\phi_j}.
\end{eqnarray}
Now, for any pair $i,j$ such that $i\geq r_l+1$ or $j\geq r_l+1$ (i.e., either $\ket{\phi_i}$ or $\ket{\phi_j}$ belongs to the kernel of $N_l$), the above condition gives us
\begin{eqnarray}
0\geq\pm\bra{\phi_i}\Omega_l^a\ket{\phi_j},
\end{eqnarray}
which directly implies that $\bra{\phi_i}\Omega_l^a\ket{\phi_j}=0$ for any pair $i,j$ such that either $i\geq r_l+1$ or $j\geq r_l+1$. This directly implies that $\Omega_l^a$ act nontrivially 
only on $\mathrm{supp}(N_l)$, and thus $\mathrm{supp}(\Omega_l^a)\subseteq \mathrm{supp}(N_l)$.
%
%%%%%%%%%%%%%%%%%%%%%%%%%%%%%%%%%%%%%%%%%%%%%%%%%%%%%%%%%%%%%%%%%%%%%%%%%%%%%%%%%%%%%%%%%%%%%%%%%%%%%%%%%%%%%%%%%%%%%%%%%%%%%%%

\subsection{Alternative proof for projective measurements}

Here we present an alternative proof in the case when we want to certify Eve's second measurement $E_1=\{R_{l|1}\}$ to be a particular reference projective measurement $E'_1=\{R'_{l|1}\}$ (up to some additional degrees of freedom). For this purpose, we use the states generated by the sources, certified as in Eq. \eqref{statest1}, and the certified measurements in Eqs. \eqref{stmea1} and \eqref{stmea2} along with some additional conditions imposed on the observed correlations.

Let us consider an ideal reference projective measurement 
$E'_1=\{R'_{l|1}\}$ where $R'_{l|1}$ are mutually orthogonal projections defined on a Hilbert space of arbitrary finite dimension $\mathbbm{C}^D$
whose ranks are in general arbitrary; let denote them by $r_l=\mathrm{rank}(R'_{l|1})$. Let us also assume
that the measurement has $K\leq 2^N$ outcomes; $r_1+\ldots+r_K=D$. Such a measurement can always be embedded in an $N$-qubit Hilbert space 
$(\mathbbm{C}^2)^{\otimes N}$ where $N$ the minimal 
natural number such that $D\leq 2^N$. In order to represent $R'_{l|1}$ in the $N$-qubit Hilbert space one can for instance assign to any element of the
standard basis $\ket{i}\in \mathbbm{C}^D$ 
an element $\ket{i_1\ldots i_N}$ of the product standard basis of $(\mathbbm{C}^2)^{\otimes}$
such that $i_1\ldots i_N$ is a binary representation of $i$ $(i=0,\ldots,D-1)$. 
Then, one needs to add another projective measurement element $M^{\perp}$ to $E_1'$ so that the resulting measurement defined on $(\mathbbm{C}^2)^{\otimes N}$ satisfies
\begin{equation}
    \sum_{i=1}^KR'_{i|1}+M^{\perp}=\mathbbm{1}_{2^N},
\end{equation}
$\mathbbm{1}_{2^N}$ is the identity acting on $(\mathbbm{C}^2)^{\otimes N}$. Given that the above procedure applies to any projective measurement, We can assume that the reference projective measurement $E_1'$ is from the beginning defined on the $N$-qubit Hilbert space.

Let us then decompose every measurement 
element in the Pauli basis \eqref{idealPauli}, 
as
\begin{eqnarray}\label{projDec}
R'_{l|1}=\sum_{i_1,\ldots,i_N=0}^{3}f^l_{i_1,\ldots,i_N}
\sigma_{1,i_1}\otimes\ldots\otimes \sigma_{N,i_N},
%\bigotimes_{k=1}^N\sigma_{k,i_k},
\end{eqnarray}
where the additional subscript marks the system on which the Pauli matrix $\sigma_{k,i_k}$ acts; for instance, $\sigma_{3,2}$ is $\sigma_2$ acting on site $3$. Then, 
$f_{i_1,\ldots,i_N}^l$ are real numbers defined as 
\begin{eqnarray}  f_{i_1,\ldots,i_N}^l=\frac{1}{2^N}\Tr[(\sigma_{1,i_1}\otimes\ldots\otimes \sigma_{N,i_N}
)R'_{l|1}].
\end{eqnarray}

Let us now move on to our self-testing scheme and consider again the quantum network scenario presented in above. Let us then assume that apart from the conditions imposed in Theorem \ref{theorem1}, the correlations observed in the network satisfy the following additional set of constraints
\begin{eqnarray}\label{betstatfullp}
\sum_{i_1,\ldots,i_N=0}^3f^l_{i_1,\ldots,i_N}\left\langle\tilde{A}_{1,i_1}\otimes\bigotimes_{k=2}^N A_{k,i_k}\otimes (R_{l|1})_E\right\rangle_{\ket{\psi_{AE}}}&=&\frac{r_l}{2^N}\quad \forall l,
\end{eqnarray}
where $\tilde{A}_{1,0}, \tilde{A}_{1,1}$ are given in 
\eqref{overA} with $\tilde{A}_{1,2}=A_{1,2}, \tilde{A}_{1,3}=\I$ and $A_{k,3}=\I$ for any $k$. 
%%\textbf{Along with it, one also needs to satisfy %the statistics for the additional outcomes %$l=K+0,\ldots,K-1+2^N-d$ as}
%
%\begin{eqnarray}\label{betstatfullp1}
%\sum_{i_1,\ldots,i_N=0}^3f^l_{i_1,\ldots,i_N}\left\%langle\tilde{A}_{1,i_1}\otimes \bigotimes_{k=2}^N 
%A_{k,i_k}\otimes 
%(R_{l|1})_E\right\rangle_{\ket{\psi_{AE}}}&=&\frac{1
%}{2^N}.
%\end{eqnarray}

We can now present our proof.

\begin{thm}\label{theorem2}
Let us suppose that the correlations generated in the quantum network satisfy the assumptions of Theorem \ref{theorem1} and the additional conditions in Eq. (\ref{betstatfullp}). Then, for any $l=0,\ldots,K-1$ we conclude that 
\begin{eqnarray}\label{stmeaNLWE1}
U_E R_{l|1}\, U_E^{\dagger} &=&(R'_{l|1})^* \otimes \mathbbm{1}_{E''},\quad \mathrm{if}\quad A_{i,2}=Y\otimes\I \quad\forall i, \quad\text{or}\quad \nonumber\\
 U_E R_{l|1}\, U_E^{\dagger} &=& R'_{l|1} \otimes \mathbbm{1}_{E''}\quad \ \ \ \mathrm{if}\quad  A_{i,2}=-Y\otimes\I \quad\forall i
\end{eqnarray}
where $U_E$ is the same unitary as in Theorem \ref{theorem1}
and $E''=E_1''\ldots E_N''$.
\end{thm}

\begin{proof}
For simplicity, we represent $R_{l|1}\equiv R_{l}$ throughout the proof. Let us first consider the relation in Eq. (\ref{betstatfullp}) for a particular $l$ and then expand it by using the fact that the observables $A_{s,j}$ for $s=1,\ldots,N$ and $j=0,1,2$ are certified as in Eqs. \eqref{stmea1} and \eqref{stmea2}. As $A_{s,2}$ for any $s$ can have two realisations as given in Eq. \eqref{stmea2}, we first consider the case when $A_{s,2}=Y\otimes\I_{s''}$. This gives us
\begin{eqnarray}\label{stcond1NLWEp}
\sum_{i_1,\ldots,i_N=0}^{3}f^l_{i_1,\ldots,i_N}\left\langle\left(\bigotimes_{s=1}^{N} U_{A_s}^{\dagger}\right) \left(\bigotimes_{k=1}^N\sigma_{k,i_k}\otimes\I_{A''}\right)\left(\bigotimes_{s=1}^{N} U_{A_s}\right)\otimes R_l\right\rangle_{\psi_{AE}}=\frac{r_l}{2^N}.
\end{eqnarray}
Then, it follows from Theorem \ref{theorem1} [cf. Eq. (\ref{A10})] that the global state can be represented in the following way 
\begin{eqnarray}\label{statep1}
\ket{\psi_{AE}}=\bigotimes_{s=1}^{N}\ket{\psi_{A_sE_s}}=\bigotimes_{s=1}^{N}( U_{A_s}^{\dagger}\otimes U_{E_s}^{\dagger})(|\phi^+_{A_s'E_s'}\rangle\otimes\ket{\xi_{A_s''E_s''}}).
\end{eqnarray}
Notice also that by virtue of Eq. \eqref{junkst1} the junk states $\ket{\xi_{A_s''E_s''}}$ can be represented as
\begin{eqnarray}\label{junkstp2}
 \ket{\xi_{A_i''E_i''}}= \left(\I_{A_i''}\otimes\sqrt{d'_i\sigma_{E_i''}}\right) \ket{\phi^{+}_{A''_iE''_i}},
 %
 %\ket{\xi_{A_s''E_s''}}= (\I_{A_s''}\otimes %P_{E_s''} )|\phi^{+}_{d_{s}''}\rangle_{A_s''E_s''},
\end{eqnarray}
where we omitted the local dimension $d_i'$ of the maximally entangled state.
Thus, the joint state can further be represented in terms of a single maximally entangled states
between $A$ and $E$ systems [cf. Eq. \eqref{MaxEntDN}] as
\begin{eqnarray}\label{115p}
\left(\bigotimes_{s=1}^{N}U_{A_s}\otimes U_{E_s}\right)|\psi_{AE}\rangle=\left(\bigotimes_{s=1}^{N} \sqrt{d'_s\sigma_{E_s''}}\right)|\phi^{+}_{D_N}\rangle_{A|E}.
\end{eqnarray}
Here the local dimension of the state is $D_N=d_1\ldots d_N=2^Nd_1'\ldots d_N'$ and we used $A|E$ to highlight the fact that $|\phi^{+}_{D_N}\rangle_{A|E}$ is a maximally entangled state between all external parties $A_i$ and the central one $E$.  

Plugging the state \eqref{115p} into Eq. \eqref{stcond1NLWEp}, and then using the fact that $(\I\otimes Q)|\phi^{+}_{d}\rangle=(Q^T\otimes\I)|\phi^{+}_{d}\rangle$ for any matrix $Q$, we obtain
\begin{eqnarray}\label{1161}
\sum_{i_1,\ldots,i_N=0}^{3}f^l_{i_1,\ldots,i_N}\left\langle \left(\bigotimes_{k=1}^N\sigma_{k,i_k}\otimes\bigotimes_{s=1}^{N}d'_s\sigma^T_{A_s''}\right)\overline{R}_l^T\otimes\I_E\right\rangle_{\phi^{+}_{D_N}}=\frac{r_l}{2^N},
\end{eqnarray}
where 
\begin{eqnarray}\label{R0eqp}
\overline{R}_l=\left(\bigotimes_{s=1}^{N} U_{E_s} \right) R_l\left(\bigotimes_{s=1}^{N} U_{E_s}^{\dagger}\right).
\end{eqnarray}
Then, tracing out the Eve's subsystems from the expectation value in Eq. (\ref{1161}), we arrive at 
\begin{eqnarray}
\sum_{i_1,\ldots,i_N=0}^{3}f^l_{i_1,\ldots,i_N}\Tr\left[ \left(\bigotimes_{k=1}^N\sigma_{k,i_k}\otimes\bigotimes_{s=1}^{N} \sigma^T_{A_s''}\right)\overline{R}_l^T\right]=r_l.
\end{eqnarray}
Using Eq. \eqref{projDec}, we can simplify the above formula as
\begin{eqnarray}\label{b18}
    \Tr\left[ \left(R'_{l|1}\otimes\bigotimes_{s=1}^{N}\sigma^T_{A_s''}\right)\overline{R}_l^T\right]=r_l.
\end{eqnarray}

Let us now concentrate on a particular outcome $l$ and consider an orthonormal basis in $\mathbbm{C}^{2^N}$ which contains the eigenvectors of the projection $R'_{l|1}$; we denote this basis by $\mathcal{B}=\{|\delta_{m}\rangle\}$, where the first $r_l$ vectors for $m=1,\ldots,r_l$
are the eigenvectors of $R'_{l|1}$. Due to the fact that $\overline{R}_0^T$ acts on the Hilbert space $(\mathbbm{C}^2)^{\otimes N}\otimes\bigotimes_{s=1}^{N}\mathcal{H}_{{A_s''}}$, we can express it using the basis $\mathcal{B}$ as
\begin{eqnarray}\label{genRp}
\overline{R}_0^T=\sum_{m,m'}\ket{\delta_{m}}\!\bra{\delta_{m'}}\otimes \tilde{R}_{m,m'},
\end{eqnarray}
where $\tilde{R}_{l,l'}$ are some unknown matrices acting on $\bigotimes_{s=1}^{N}\mathcal{H}_{{A_s''}}$. 

After plugging Eq. (\ref{genRp}) into (\ref{b18})
and using the cyclic property of trace, we then obtain
\begin{eqnarray}\label{nlwest11p}
   \sum_{m=1}^{r_{l}}   \Tr\left[ \left(\bigotimes_{s=1}^{N}\sigma^T_{A_s''}\right)\tilde{R}_{m,m}\right]=r_l.
\end{eqnarray}
Since each term under the above sum is upper-bounded by one, we have that 
\begin{equation}
\Tr\left[ \left(\bigotimes_{s=1}^{N}\sigma_{A_s''}^T\right)\tilde{R}_{m,m}\right]=1\qquad (m=1,\ldots,r_l).    
\end{equation}
Now, due to the fact that the density matrices $\sigma_{A''_s}$ are full rank and $0\leqslant \tilde{R}_{m,m}\leqslant \mathbbm{1}$ which stems from the fact that $\{R_l\}$ is a quantum measurement, one observers that the above equation holds true if, and only if, $\tilde{R}_{m,m}=\mathbbm{1}$ for any $m=1,\ldots,r_0$. This implies that 
\begin{eqnarray}\label{121}
\overline{R}_l^T&=&R'_{l|1}\otimes\I_{A''}+\mathbbm{L}_l,
\end{eqnarray}
where $\mathbbm{L}_l$ is an operator composed of the remaining terms of 
$\overline{R}_l$ appearing in the decomposition (\ref{genRp}); in other words, $\mathbbm{L}_l=\overline{R}_l^T-R'_{l|1}\otimes \mathbbm{1}$.

\begin{eqnarray}\label{L03}
\mathbbm{L}_l&=&\sum_{\substack{m,m'=1\\m\neq m'}}^{r_l}\ket{\delta_m}\!\!\bra{\delta_{m'}}\otimes \tilde{R}_{m,m'}+\sum_{m=1}^{r_l}\sum_{m'=r_l+1}^{2^N}\ket{\delta_m}\!\!\bra{\delta_{m'}}\otimes \tilde{R}_{m,m'}\nonumber\\
&&+\sum_{m=r_l+1}^{2^N}\sum_{m'=1}^{r_l}\ket{\delta_m}\!\!\bra{\delta_{m'}}\otimes \tilde{R}_{m,m'}+\sum_{m=r_{l}+1}^{2^N}\sum_{m'=r_l+1}^{2^N}\ket{\delta_m}\!\!\bra{\delta_{m'}}\otimes \tilde{R}_{m,m'}.
\end{eqnarray}
Let us now show that the fact that $\overline{R}_l^T\leq \mathbbm{1}$ imposes that the 
first three sums in the above representation must vanish. To this aim, let us consider a product vector
$\ket{\delta_m}\ket{\xi}$ for any $m=1,\ldots,r_l$
and any vector $\ket{\xi}$, and act on it with $\overline{R}_l^T=R'_{l|1}\otimes\I_{A''}+\mathbbm{L}_l$.
This gives
\begin{equation}
\overline{R}_l^T\ket{\delta_m}\ket{\xi}=\ket{\delta_m}\ket{\xi}+\sum_{\substack{m'=1\\m'\neq m}}^{2^N}\ket{\delta_{m'}}\otimes \tilde{R}_{m',m}\ket{\xi},
\end{equation}
which further leads to the following condition
\begin{equation}
    \bra{\delta_m}\bra{\xi}(\overline{R}_l^T)^2\ket{\delta_m}\ket{\xi}=1+\sum_{\substack{m'=1\\m'\neq m}}^{2^N} \bra{\xi}\tilde{R}_{m',m}^{\dagger}\tilde{R}_{m',m}\ket{\xi}.
\end{equation}
Let us observe now that $\overline{R}_l^T\leq \mathbbm{1}$ implies $(\overline{R}_l^T)^2\leq \mathbbm{1}$, which after applying to the above equation allows one to conclude that
\begin{equation}
    \sum_{\substack{m'=1\\m'\neq m}}^{2^N} \bra{\xi}\tilde{R}_{m',m}^{\dagger}\tilde{R}_{m',m}\ket{\xi}\leq 0.
\end{equation}
Therefore, $\bra{\xi}\tilde{R}_{m',m}^{\dagger}\tilde{R}_{m',m}\ket{\xi}=0$ for any vector $\ket{\xi}$, which in turn means that $\tilde{R}_{m',m}=0$ for any $m'\neq m$.
Since the same argument can be repeated for any $\ket{\delta_m}$ with $m=1,\ldots,r_l$ one then has that all the matrices 
$\tilde{R}_{m',m}=0$
for any $m=1,\ldots,r_l$ and any $m'\neq m$.
This finally implies that the $\mathbbm{L}_l$ operator
can be rewritten as 
\begin{eqnarray}\label{L04}
\mathbbm{L}_l=\sum_{m=r_{l}+1}^{2^N}\sum_{m'=r_l+1}^{2^N}\ket{\delta_m}\!\bra{\delta_{m'}}\otimes \tilde{R}_{m,m'},
\end{eqnarray}
and thus it acts on a subspace of the corresponding Hilbert space which is orthogonal to that supporting
$R'_{l|1}\otimes \mathbbm{1}_{A''}$. Due to the fact that 
the measurement operator satisfies $\overline{R}_l\geq 0$ the latter clearly implies
that $\mathbbm{L}_l\geq 0$. 

It is important to stress that the same reasoning can be applied to 
$\overline{R}_l^T$ for any $l$, and thus 
every $\overline{R}_l^T$ decomposes into a direct sum 
\begin{equation}\label{costam}
\overline{R}_l^T=R'_{l|1}\otimes\I_{A''}+\mathbbm{L}_l,
\end{equation}
of two positive semi-definite operators $R'_{l|1}\otimes\I_{A''}$ and $\mathbbm{L}_l$.
Now, after summing \eqref{costam} over all outcomes, 
one obtains
\begin{equation}\label{costam2}
    \sum_{l}\mathbbm{L}_l=0,
\end{equation}
where we also exploited the fact that both $\overline{R}_l$ and $R'_{l|1}$ sum up to the identities. 
Taking into account that $\mathbbm{L}_l\geq 0$, one deduces from the condition (\ref{costam2}) that 
$\label{L}_l=0$ for any $l$ and thus
\begin{equation}
    \overline{R}_l^T=R'_{l|1}\otimes\I_{A''}
\end{equation}
or, equivalently, 
\begin{equation}
    \overline{R}_l=(R'_{l|1})^{*}\otimes\I_{A''}.
\end{equation}

Similarly, one can consider the other case of $A_{s,2}=-Y\otimes\I_{s''}$ for any $s$ and obtain an analogous relation to Eq. \eqref{b18} which reads
\begin{eqnarray}
    \Tr\left[ \left((R'_{l|1})^{*}\otimes\bigotimes_{s=1}^{N}\sigma^T_{A_s''}\right)\overline{R}_l^T\right]=r_l,\qquad (l=0,\ldots,K-1).
\end{eqnarray}
Following exactly the same steps as above, we find that $\overline{R}_l=R'_{l|1}\otimes\I$ for $l=0,\ldots,K-1$, thus completing the proof.
\end{proof}

\end{proof}

\section{Appendix C: Self-testing any quantum state}
Using the Theorems \ref{theorem3} proven above, we can now show how it is possible to probabilistically self-test arbitrary pure or mixed $N$-qubit quantum state. For this purpose, let us first certify the post-measurement state when eve chooses $e=1$ and obtains an outcome $l$.
\setcounter{thm}{0}
\begin{cor}\label{corr1}
Assume that the states are certified as in Eq. \eqref{statest1} and Eve's measurement $E_1$ is certified as in Eq. \eqref{stmeaNLWE2}. Consequently, when Eve observes the $l-$th outcome of her measurement, the post-measurement state with the external parties is given by
\begin{eqnarray}
U_A\,\rho_{A}^l\,U_A^{\dagger}&=&\frac{1}{\Tr R'_{l|1}}R'_{l|1}\otimes \tilde{\rho}_{A''},\quad \mathrm{if}\quad A_{i,2}=Y\otimes\I \quad\forall i, \quad\text{or}\quad \nonumber\\
 U_A\,\rho_{A}^l\,U_A^{\dagger}&=&\frac{1}{\Tr R'_{l|1}}(R'_{l|1})^*\otimes \tilde{\rho}_{A''}\quad \ \ \ \mathrm{if}\quad  A_{i,2}=-Y\otimes\I \quad\forall i,
\end{eqnarray}
where $U_A=\bigotimes_{s=1}^{N}U_{A_s}$ and the unitaries $U_{A_s}$ are the same as in Eq. 
\eqref{statest1}.
\end{cor}
\begin{proof}
The post-measurement state when Eve observes the $l-$th outcome of her measurement $E_1$ is given by
    \begin{eqnarray}
        \rho_{A}^l=\frac{1}{\overline{P}(l|e=1)}\Tr_{E}\left[\left(\I_{A}\otimes R_{l|1}\right)\bigotimes_{s=1}^{N}\proj{\psi_{A_sE_s}}\right].
    \end{eqnarray}
Then, substituting the states $\ket{\psi_{A_sE_s}}$ from Eq. \eqref{A10} and the measurement elements $R_{l|1}$ from Eq. \eqref{statest1}, we get that
    \begin{eqnarray}
        U_A\,\rho_{A}^l\,U_A^{\dagger}=\frac{1}{\overline{P}(l|e=1)}\Tr_{E}\left[\left(\I_{A}\otimes (R'_{l|1})^*\otimes\I_{E''}\right)\bigotimes_{s=1}^{N}\proj{\phi^+}_{A_s'E_s'}\otimes \proj{\xi_s}_{A_s''E_s''}\right].
    \end{eqnarray}
Again using the identity $(\I\otimes Q)\ket{\phi^{+}}=(Q^T\otimes\I)\ket{\phi^{+}}$, we obtain
\begin{eqnarray}
U_A\,\rho_{A}^l\,U_A^{\dagger}=\frac{1}{\overline{P}(l|e=1)}\Tr_{E}\left[\left(R'_{l|1}\otimes\mathbbm{1}_{A''}\otimes\I_{E}\right)\bigotimes_{s=1}^{N}\proj{\phi^+}_{A_s'E_s'}\otimes \proj{\xi_s}_{A_s''E_s''}\right],
\end{eqnarray}
where we also used the fact that 
$\left((R'_{l|1})^*\right)^T=R'_{l|1}$. After tracing the $E$ subsystem we arrive at
\begin{eqnarray}
    U_A\,\rho_{A}^l\,U_A^{\dagger}=\frac{1}{2^N\overline{P}(l|e=1)} R'_{l|1}\otimes \tilde{\rho}_{A''},
\end{eqnarray}
where $\tilde{\rho}_{A''}=\Tr_{E''}\left[\bigotimes_{s=1}^{N}\proj{\xi_s}_{A_s''E_s''}\right]$. It is straightforward to observe now that $\overline{P}(l|e=1)=\Tr\left(R'_{l|1}\right)/2^N$ giving us the desired result.
\end{proof}

Now, to certify an arbitrary pure state $\ket{\psi}\in(\mathbbm{C}^2)^{\otimes N}$ among the external parties it is enough that the second Eve's reference measurement $E_1'$ is projective and one of its elements, say $R'_{0|1}$, is a projection onto $\ket{\psi}$, $R'_{0|1}=\proj{\psi}$.

For mixed states the situation is a bit more complicated. Consider a mixed state acting on $\mathbbm{C}^d$: 
\begin{equation}
 \rho=\sum_kp_k\proj{\psi_k}. 
\end{equation}
It follows that one can construct an extremal $3d-$outcome POVM in the Hilbert space $\mathbbm{C}^{2d}$ that can be performed by Eve on her share of the joint state $\ket{\psi_{AE}}$ to create $\rho$ at the external parties' labs with the aid of post-processing. To construct this POVM, we first define two sets of $d$ mutually orthogonal vectors $\mathcal{B}_1=\{\ket{\psi_k}\}$ and  $\mathcal{B}_2=\{\ket{\phi_k}\}$ such that $\bra{\psi_k}\phi_{k'}\rangle=0$ for any $k,k'$. Then, one considers the pair of the states $\{\psi_k,\phi_k\}$ to construct the trine POVM that acts on $\mathbbm{C}^2$ [cf. Ref. \cite{remik1}] as 
\begin{eqnarray}
    M_{k,1}=p_k \proj{\psi_k},\qquad M_{k,2}=\frac{2-p_k}{2}\proj{\tau_{k,2}},\qquad M_{k,3}=\frac{2-p_k}{2}\proj{\tau_{k,3}},
\end{eqnarray}
where for all $k$,
\begin{eqnarray}\label{tau}
    \ket{\tau_{k,2}}=\sqrt{\frac{1-p_k}{2-p_k}}\ket{\psi_k}+\sqrt{\frac{1}{2-p_k}}\ket{\phi_k},\qquad \ket{\tau_{k,3}}=-\sqrt{\frac{1-p_k}{2-p_k}}\ket{\psi_k}+\sqrt{\frac{1}{2-p_k}}\ket{\phi_k}.
\end{eqnarray}

Notice that for any $k$, 
\begin{equation}\label{D8}
M_{k,0}+M_{k,1}+M_{k,2}=\proj{\psi_k}+\proj{\phi_k},\end{equation}
and therefore the three-element set $\{M_{k,l}\}_{l}$ can also be understood as a three-outcome rank-one POVM in $\mathbbm{C}^2$. Moreover, taking a sum over $k$ in Eq. (D8) one obtains that
\begin{equation}
\sum_{k=0}^{d-1}(M_{k,0}+M_{k,1}+M_{k,2})=\sum_{k=0}^{d-1}(\proj{\psi_k}+\proj{\phi_k})=\mathbbm{1}_{\mathbbm{C}^{2d}}
\end{equation}
which implies that $M_{k,l}$ form a valid $3d$-outcome rank-one POVM in $\mathbbm{C}^{2d}$ which we denote $\mathcal{M}=\{M_{k,l}\}$. 

Let us now demonstrate that $\mathcal{M}$ is an extremal POVM. To this end, we first notice that for any $k$ the elements $M_{k,l}$ $(l=0,1,2)$ are linearly independent and, moreover, the sets $\{M_{k,0},M_{k,1},M_{k,2}\}$ and $\{M_{k',0},M_{k',1},M_{k',2}\}$ for any pair $k\neq k'$ are orthogonal in the sense that any element from the first set is orthogonal to any element in the second set. This means that all $M_{k,l}$ forming our POVM are linearly independent which together with the fact that $M_{k,l}$ are rank-one
implies that $\mathcal{M}$ is extremal [cf. Corollary 5 in Ref. \cite{APP05}].

Let us now go back to our scheme of probabilistic self-testing of an arbitrary mixed state $\rho$ and observe that after performing $\mathcal{M}$ on her shares of the joint state, Eve obtains the outcomes $l\equiv (k,1)$ with probability $p(k,1|e=1)p_l/2^N$. Whenever she obtains the outcome $(k,1)$, the post-measurement state with the external parties is certified from Corrollary \ref{corr1}. Now, one can clearly see that as the average state with the external parties when Eve obtains $(k,1)$ for all $k$ is certified from Corrollary \ref{corr1} to be
\begin{eqnarray}
   U_A\,\rho_{A}\,U_A^{\dagger}=\frac{1}{\sum_kp(k,1)}\sum_{k}p(k,1) U_A\,\rho_{A}^{(k,1)}\,U_A^{\dagger}=\left(\sum_{k}M_{(k,1)}\right)\otimes \tilde{\rho}_{A''}=\rho_{A'}\otimes \tilde{\rho}_{A''},
\end{eqnarray}
or its complex conjugate, which is exactly as we promised. The above state occurs with a probability $1/2^N$.

\section{Appendix D: Robustness analysis of self-testing any extremal measurement}
Let us now find the robustness of self-testing the extremal measurements. Eve's measurement corresponding to $e=0$ can be robustly self-tested to be the GHZ-like basis \eqref{statest1} using the direct generalisation of the proof in \cite{sarkarPRL}, which utilises a robustness result in \cite{sarkar2025}. Thus, we directly proceed towards robust self-testing of the measurement corresponding to Eve's input $e=1$. For this purpose, we again need to utilise the robust self-testing in the star network as stated in \cite{sarkar2025} and the relevant part for our work is stated below.
\setcounter{fakt}{0}
\begin{fakt}\label{fact1}
   Suppose that the the Bell functionals $\langle\hat{\mathcal{I}}_{l}\rangle$ \eqref{BE1Nop} attain a value $\varepsilon-$close to their maximal quantum values, that is,
\begin{eqnarray}
 \langle\psi_l|\hat{\mathcal{I}}_{l}|\psi_l\rangle\geq \beta_Q-\varepsilon\qquad \forall l.
\end{eqnarray}
along with the probability of the outcomes of the measurement of Bob being $|\overline{P}(l)-1/2^N|\leq\varepsilon$  for all $l$. 
Then, there exists unitaries $U_i:\mathcal{H}_{A_i}\rightarrow\mathcal{H}_{A_i}$
for $i=1,\ldots,N$ such that the state is certified as
\begin{eqnarray}\label{Rob1manu1}
 \dl{\ket{\tilde{\psi}_l}-\ket{\phi_l}\ket{\xi_l}}\leq \left(\delta_N+\sqrt{2(N-1)}\right)\sqrt{2\varepsilon}\leq 17(N^2-1)\sqrt{2\varepsilon}.
\end{eqnarray}
where $\ket{\tilde{\psi}_l}=\bigotimes_{i=1}^NU_{i}\ket{\psi_l}$ and $\delta_N=16(N-1)+ 2N(N-1)\left[\sqrt{2}+1+\sqrt{\frac{1}{N-1}}\right]\leq 16(N^2-1)$, while the measurements are characterised as
\begin{eqnarray}\label{roburel11}
     \dl{U_1\tilde{A}_{1,0}U_1^{\dagger}\ket{\tilde{\psi}_l}-Z\otimes\I \ket{\tilde{\psi}_l}}\leq 8\sqrt{2\varepsilon},\qquad \dl{U_1\tilde{A}_{1,1}U_1^\dagger\ket{\tilde{\psi}_l}-X\otimes\I \ket{\tilde{\psi}_l}}\leq 8\sqrt{2\varepsilon}.
\end{eqnarray}
and,
\begin{eqnarray}\label{roburel12}
    U_iA_{i,0}U_i^\dagger\ket{\tilde{\psi}_l}=Z\otimes\I\ket{\tilde{\psi}_l},\qquad\dl{U_iA_{i,1}U_i^\dagger\ket{\tilde{\psi}_l}-X\otimes\I\ket{\tilde{\psi}_l}}\leq \left[\sqrt{2}+1+\sqrt{\frac{1}{N-1}}\right]\sqrt{2\varepsilon}\leq 4\sqrt{2\varepsilon}.
\end{eqnarray}
\end{fakt}
Let us prove some additional robustness statements following from the above fact required for our purpose.

\begin{cor}
    If the conditions of Fact \eqref{fact1} is satisfied, then the states $\ket{\psi_{A_iE_i}}$ generated by the sources in the star network scenario depicted in Fig. 1 of the manuscript are robustly certified to be
    \begin{eqnarray}\label{robst1}
         \dl{\ket{\overline{\psi}_{A_iE_i}}-\ket{\phi^+_{A_i'E'_i}}\ket{\xi_{A_i''E''_i}}}\leq N^2\left(\delta_N+\sqrt{2(N-1)}\right)^2\sqrt{\varepsilon/2}\leq (17N(N^2-1))^2\sqrt{\varepsilon/2}.
    \end{eqnarray}
where $\ket{\overline{\psi}_{A_iE_i}}=U_{A_i}\otimes U_{E_i}\ket{\tilde{\psi}_{A_iE_i}}$ for all $i$. Moreover, the observable $A_{i,2}$ is also robustly certified to be
\begin{eqnarray}\label{Yrob}
    \dl{U_iA_{i,2}U_i^\dagger\ket{\tilde{\psi}_l}-Y\otimes\mathcal{K}\ket{\tilde{\psi}_l}}\leq v_N\sqrt{2\varepsilon} \leq 226 N^2\sqrt{2\varepsilon}
\end{eqnarray}
where $v_N=12\gamma_N+2(\delta_N+\sqrt{2(N-1)})$ and $\gamma_N=\sqrt{2}+\delta_N+(N-2)(\sqrt{2}+1+\sqrt{1/(N-1)})+\sqrt{2(N-1)}\leq 16 N^2+10N$ along with the operator $\mathcal{K}$ satisfying $\mathcal{K}^2\leq \I, \mathcal{K}=\mathcal{K}^{\dagger}$ and satisfying $\dl{\mathcal{K}\bigotimes_i\ket{\xi_{A''_iE''_i}}}^2\geq 1-8\gamma_N^2\varepsilon$.
\end{cor}
\begin{proof}
    Let us begin by proving the robustness of the states generated by the sources. For this, we consider Eq. \eqref{Rob1manu1} and find that
    \begin{eqnarray}
         \dl{\proj{\tilde{\psi}_l}-\proj{\phi_l}\otimes\proj{\xi_l}}_*\leq N^2\left(\delta_N+\sqrt{2(N-1)}\right)^2\sqrt{\varepsilon/2}
    \end{eqnarray}
    where $\dl{M}_*=\Tr|M|$.
    One can also look at Eqs. $(68)-(75)$ in \cite{sarkar2025} for a more rigorous proof of the above statement. Now, recalling that $\proj{\tilde{\psi}_l}$ is the post-measurement state corresponding to Eve's measurement element $R_{l|0}$ and thus using the formula \eqref{54}, we obtain
    \begin{eqnarray}
       \dl{  \frac{1}{\overline{P}(l|0)}\Tr_{E}\left[\left(\I_{A}\otimes R_l\right)\bigotimes_{i=1}^N\proj{\psi_{A_iE_i}}\right]-\proj{\phi_l}\otimes\proj{\xi_l}}_*\leq N^2\left(\delta_N+\sqrt{2(N-1)}\right)^2\sqrt{\varepsilon/2}.
    \end{eqnarray}
    Proceeding now in the same manner as Eqs. \eqref{Schmidt}-\eqref{cond4}, %and using the fact that $\overline{P}(l|0)\leq1$, 
    we have from the above formula that
\begin{eqnarray}
       \dl{  \left(\bigotimes_{i=1}^NP_{A_i}U_{{i}}\right)\ R_l\ \left(\bigotimes_{i=1}^NU_{{i}}^{\dagger} P_{{A}_i}\right)-\overline{P}(l|0)\proj{\phi_l}\otimes\proj{\xi_l}}_*\leq \overline{P}(l|0)N^2\left(\delta_N+\sqrt{2(N-1)}\right)^2\sqrt{\varepsilon/2}.
    \end{eqnarray}
    Recall here that $P_{A_i}$ are the state parameters \eqref{P}. Again using triangle inequality in the above formula, %and the fact that $|\overline{P}(l|0)-1/2^N|\leq \varepsilon$, 
    we have that
    \begin{eqnarray}
        \dl{  \left(\bigotimes_{i=1}^NP_{A_i}U_{{i}}\right)\ R_l\ \left(\bigotimes_{i=1}^NU_{{i}}^{\dagger} P_{{A}_i}\right)-\overline{P}(l|0)\proj{\phi_l}\otimes\proj{\xi_l}}_*\leq \overline{P}(l|0)N^2\left(\delta_N+\sqrt{2(N-1)}\right)^2\sqrt{\varepsilon/2}.
    \end{eqnarray}
   Summing over $l$ and then using triangle inequality, we obtain
   \begin{eqnarray}
        \dl{  \bigotimes_{i=1}^NP_{A_i}^2-\sum_{l=0}^{2^N-1}\overline{P}(l|0)\proj{\phi_l}\otimes\proj{\xi_l}}_*\leq N^2\left(\delta_N+\sqrt{2(N-1)}\right)^2\sqrt{\varepsilon/2}.
   \end{eqnarray}
   Now, taking a partial trace over $A_1,\ldots,A_N/A_k$ and using the fact that $\dl{\Tr_{A}M_{AB}}_*\leq \dl{M_{AB}}_*$, we obtain that
   \begin{eqnarray}
        \dl{  P_{A_k}^2-\I_{A_i'}\otimes\frac{\sigma_{A_k''}}{2} }_*\leq N^2\left(\delta_N+\sqrt{2(N-1)}\right)^2\sqrt{\varepsilon/2}
   \end{eqnarray}
   where $\sigma_{A_k''}$ for all $k$ is similar to Eq. \eqref{sigma1} as $\sigma_{A_j''}=\sum_{l=0}^{2^N-1}\overline{P}(l|0)\Tr_{A''\setminus\{A_{j}''\}}\left(\proj{\xi_l}_{A''}\right)$. Notice that $P\geq0$ and $\dl{M}\leq\dl{M}_*$ and thus, we obtain from the above formula 
   \begin{eqnarray}
        \dl{  P_{A_k}-\I_{A_i'}\otimes\sqrt{\frac{\sigma_{A_k''}}{2} }}\leq N^2\left(\delta_N+\sqrt{2(N-1)}\right)^2\sqrt{\varepsilon/2}.
   \end{eqnarray}
   Now, using the fact that $\bra{\phi_+^d}M^{\dagger}M\otimes\I\ket{\phi_d^+}=\dl{M}/d$, we obtain that
   \begin{eqnarray}
         \dl{ \sqrt{d_i} P_{A_k}\ket{\phi_{d_i}^+}-\I_{A_i'}\otimes\sqrt{\frac{d_i\sigma_{A_k''}}{2} }\ket{\phi_{d_i}^+}}\leq N^2\left(\delta_N+\sqrt{2(N-1)}\right)^2\sqrt{\varepsilon/2}
   \end{eqnarray}
   which using \eqref{genstate3} can be expressed as
   \begin{eqnarray}
        \dl{ \ket{\overline{\psi}_{A_iE_i}}-\ket{\phi^+_{A_i'E_i'}}\otimes\ket{\xi_{A_i''E_i''}}}\leq N^2\left(\delta_N+\sqrt{2(N-1)}\right)^2\sqrt{\varepsilon/2}.
   \end{eqnarray}
   To obtain $\I_{A_i'}\otimes\sqrt{\frac{d_i\sigma_{A_k''}}{2}}\ket{\phi_{d_i}^+} =\ket{\phi^+_{A_i'E_i'}}\otimes\ket{\xi_{A_i''E_i''}}$ we used Eq. \eqref{66} and then embedded the local Hilbert space of $A_i$ into an even dimensional Hilbert space of larger dimension. 
   
   Let us now find the robustness of the observables $A_{i,2}$. For this purpose, we utilise the sum of squares decomposition \eqref{SOS1}, which from  $\langle\psi_l|\hat{\mathcal{I}}_{l}|\psi_l\rangle\geq \beta_Q-\varepsilon$ for all $l$, allows us to conclude that
   \begin{eqnarray}
       \dl{Q_{i,l_i}\ket{\psi_l}}\leq \sqrt{2\varepsilon}
   \end{eqnarray}
   for all $i$ and $l_i=0,1$. Now, expanding $Q_{i,l_i}$ using Eq. \eqref{SOS2_2}, we obtain
   \begin{eqnarray}
       \dl{\left((-1)^{l_i}\I+A_{1,2}\otimes A_{i,2}\otimes\bigotimes_{\substack{j=2\\j\ne i}}^{N} A_{j,1}\right)\ket{\psi_l}}\leq \sqrt{2\varepsilon}
   \end{eqnarray}
   Furthermore, denoting $\overline{A}_{i,2}=U_iA_{i,2}U_i^{\dagger}$ and then using \eqref{Rob1manu1} and \eqref{roburel12}, we obtain
   \begin{eqnarray}
        \dl{\left((-1)^{l_i}\I+\overline{A}_{1,2}\otimes \overline{A}_{i,2}\otimes\bigotimes_{\substack{j=2\\j\ne i}}^{N} X_j\right)\ket{\phi_l}\ket{\xi_l}}\leq \sqrt{2\varepsilon}+(N-2)\left[\sqrt{2}+1+\sqrt{\frac{1}{N-1}}\right]\sqrt{2\varepsilon}+\left(\delta_N+\sqrt{2(N-1)}\right)\sqrt{2\varepsilon}.\nonumber\\
   \end{eqnarray}
   From here on, we denote $\gamma_N=\sqrt{2}+\delta_N+(N-2)(\sqrt{2}+1+\sqrt{1/(N-1)})+\sqrt{2(N-1)}$.
   As the observables $A_{i,2}$ act on even-dimensional local Hilbert spaces in the above formula, we express them here in the same way as \eqref{genmea2} 
   \begin{eqnarray}\label{174}
    \overline{A}_{i,2}=\sum_{j=0}^3\sigma_j\otimes \mathcal{K}_{i,j}
   \end{eqnarray}
   where $\sigma_0=\I, \sigma_1=Z,\sigma_2=X, \sigma_3=Y$ with $\mathcal{K}_{i,j}$ acting on the junk part of the state $\ket{\xi_l}$ and satisfies $\sum_{j=0}^3\mathcal{K}_{i,j}^2=\I$ and $\mathcal{K}_{i,j}=\mathcal{K}_{i,j}^{\dagger}$.
   From the above formula, we obtain
   \begin{eqnarray}
        \dl{\left(-\sigma_{3}\otimes \sigma_{3}\otimes\I+\sum_{j_1,j=0}\sigma_{j_1}\otimes\sigma_{j}\otimes\mathcal{K}_{1,j_1}\otimes\mathcal{K}_{i,j}\right)\otimes\bigotimes_{\substack{j=2\\j\ne i}}^{N} X_j\ket{\phi_l}\ket{\xi_l}}\leq \gamma_N\sqrt{2\varepsilon}
   \end{eqnarray}
where we used the fact that $(-1)^{l_i}\sigma_{3}\otimes \sigma_{3}\otimes\bigotimes_{\substack{j=2\\j\ne i}}^{N} X_j\ket{\phi_l}=-\ket{\phi_l}$. Expanding the above formula, we obtain
\begin{eqnarray}
    \sum_{\substack{j_1,j=0\\j= j_1\ne 3}}^3\dl{\mathcal{K}_{1,j_1}\otimes\mathcal{K}_{i,j}\ket{\xi_l}}+\dl{\I-\mathcal{K}_{1,3}\otimes\mathcal{K}_{i,3}\ket{\xi_l}}\leq \gamma_N\sqrt{2\varepsilon}.
\end{eqnarray}
This implies from the above expression that
\begin{eqnarray}\label{177}
    \bra{\xi_l}\mathcal{K}_{1,j_1}^2\otimes\mathcal{K}_{i,j}^2\ket{\xi_l}\leq2\gamma_N^2\varepsilon,\quad \forall j,j_1-\{j=j_1\ne3\}, \qquad\bra{\xi_l}(\I-\mathcal{K}_{1,3}\otimes\mathcal{K}_{i,3})^2\ket{\xi_l}\leq 2\gamma_N^2\varepsilon.
\end{eqnarray}
Now, let us observe from the left-hand expression in the above formula \eqref{177} that for any $j_1$ except $j_1= 3$
\begin{eqnarray}
\sum_{j=0}^3\bra{\xi_l}\mathcal{K}_{1,j_1}^2\otimes\mathcal{K}_{i,j}^2\ket{\xi_l}\leq8\gamma_N^2\varepsilon
\end{eqnarray}
which using the fact that $\sum_{j=0}^3\mathcal{K}_{i,j}^2=\I$, we obtain for $j_1=0,1,2$
\begin{eqnarray}
\bra{\xi_l}\mathcal{K}_{1,j_1}^2\ket{\xi_l}\leq8\gamma_N^2\varepsilon%\implies \dl{\mathcal{K}_{1,j_1}\ket{\xi_l}}
\end{eqnarray}
which further implies that
\begin{eqnarray}
\bra{\xi_l}\mathcal{K}_{1,3}^2\ket{\xi_l}\geq1-8\gamma_N^2\varepsilon.
\end{eqnarray}
Similarly, for any other $i=2,\ldots,N$, we have that
\begin{eqnarray}\label{180}
    \bra{\xi_l}\mathcal{K}_{i,j}^2\ket{\xi_l}\leq8\gamma_N^2\varepsilon,\quad (j=0,1,2)\qquad\bra{\xi_l}\mathcal{K}_{i,3}^2\ket{\xi_l}\geq1-8\gamma_N^2\varepsilon.
\end{eqnarray}
Notice that $\ket{\xi_l}=\ket{\xi_{A''B''}}$ for any $l$ as for the ideal Eve's measurement, $\I$ acts on the junk subsystems and thus we can also conclude from the above expression that
\begin{eqnarray}\label{181}
   \bra{\xi}\mathcal{K}_{i,3}^2\ket{\xi}\geq 1-8\gamma_N^2\varepsilon.
\end{eqnarray}

Finally, let us consider the expression $\dl{(\overline{A}_{i,2}-Y\otimes \mathcal{K}_{i,3})\ket{\tilde{\psi}_l}}$ and use triangle inequality to get
\begin{eqnarray}\label{183}
    \dl{(\overline{A}_{i,2}-Y\otimes \mathcal{K}_{i,3})\ket{\tilde{\psi}_l}}\leq \dl{(\overline{A}_{i,2}-Y\otimes \mathcal{K}_{i,3})(\ket{\tilde{\psi}_l}-\ket{\phi_l}\ket{\xi_l})}+\dl{(\overline{A}_{i,2}-Y\otimes \mathcal{K}_{i,3})\ket{\phi_l}\ket{\xi_l}}.
\end{eqnarray}
The first expression in the right-hand side of the above formula can be upper bounded using Eq. \eqref{Rob1manu1}. Thus, let us now bound the quantity $\dl{(\overline{A}_{i,2}-Y\otimes \mathcal{K}_{i,3})\ket{\phi_l}\ket{\xi_l}}$ which on expanding using \eqref{174} and then using \eqref{180} gives us
\begin{eqnarray}
  \dl{(\overline{A}_{i,2}-Y\otimes \mathcal{K}_{i,3})\ket{\phi_l}\ket{\xi_l}}\leq \sum_{j=0,1,2}\dl{\mathcal{K}_{i,j}\ket{\xi_l}}\leq 12\gamma_N\sqrt{2\varepsilon}
\end{eqnarray}
which along with \eqref{183} gives us the formula \eqref{Yrob}. This completes the proof.

\end{proof}

Let us now find the robustness of the self-testing statement corresponding to Eve's input $e=1$.

\setcounter{thm}{3}
\begin{thm}\label{thTeo5}  Consider again the network scenario outlined in the main text (see Fig. 1) with Eve's measurement $E_1=\{R_{l|1}\}$ being close to the ideal one as 
       \begin{eqnarray}\label{robures2}
        \left\|\ U_E\ R_{l|1}\bigotimes_{i}^N\ket{\psi_{A_iE_i}}-\, \left(\proj{\phi_l}_{E'}\otimes\I_{E''}\right)\ \bigotimes_{i=1}^N\ket{\phi^+_{A_i'E'_i}}\ket{\xi_{A_i''E''_i}}\right\|\leq\varepsilon\qquad \forall l
    \end{eqnarray}
    where $\ket{\psi_{A_iE_i}}, \{R_{l|1}\}$ are the actual state and measurement in the experiment.
    Then the correlations \eqref{betstatfull} are close to the ideal correlations for all $l,i_1,\ldots,i_N$ as
    \begin{eqnarray}
\left|\left\langle\tilde{A}_{1,i_1}\otimes\bigotimes_{k=2}^N A_{k,i_k}\otimes (R_{l|1})_E\right\rangle_{\psi_{AE}}-f^l_{i_1,\ldots,i_N}\right|\leq(1+4N\gamma_N)\sqrt{2\varepsilon}+f_1(\varepsilon)+(N-n_3)8\sqrt{2\varepsilon}+ n_3v_N\sqrt{2\varepsilon}
\end{eqnarray}
    where $n_3$ denotes the number of indices $i_k=3$ and $f_1(\varepsilon)=\sqrt{\varepsilon/2}(17N(N^2-1))^2$. 
\end{thm}

\begin{proof}
    Let us consider the condition \eqref{robures2} and using the fact that $\mathcal{A}_{i,j}\otimes\I_{A_i''}$ are unitary, where $\mathcal{A}_{i,j}$ are the ideal observables as stated in \eqref{GHZObs}, we obtain that
    \begin{eqnarray}
        \left\|\ \tilde{\mathcal{A}}_{1,i_1}\otimes\bigotimes_{k=2}^N \mathcal{A}_{k,i_k}\otimes\Gamma_{A''}\otimes\left( U_E\ R_{l|1}\bigotimes_{i}^N\ket{\psi_{A_iE_i}}-\, \left(\proj{\phi_l}_{E'}\otimes\I_{E''}\right)\ \bigotimes_{i=1}^N\ket{\phi^+_{A_i'E'_i}}\ket{\xi_{A_i''E''_i}}\right)\right\|\leq\sqrt{2}\varepsilon\qquad \forall l
    \end{eqnarray}
    which on using Cauchy-Schwarz inequality, gives us
    \begin{eqnarray}\label{158}
    \left|\bra{\phi^+_{A'E'}}\bra{\xi_{A''E''}}\tilde{\mathcal{A}}_{1,i_1}\otimes\bigotimes_{k=2}^N \mathcal{A}_{k,i_k}\otimes\Gamma_{A''}\otimes (R_{l|1})_E\ket{\psi_{AE}}-f^l_{i_1,\ldots,i_N}\beta\right|\leq\sqrt{2}\varepsilon
    \end{eqnarray}
    where $\bra{\phi^+_{A'E'}}=\bigotimes_{i=1}^N\bra{\phi^+_{A'_iE'_i}},\bra{\xi_{A''E''}}=\bigotimes_{i=1}^N\bra{\xi_{A''E''}}$ and we used the fact that 
    \begin{eqnarray}
f^l_{i_1,\ldots,i_N}= \bra{\phi^+_{A'E'}}\left(\tilde{\mathcal{A}}_{1,i_1}\otimes\bigotimes_{k=2}^N \mathcal{A}_{k,i_k}\otimes\proj{\phi_l}_{E'}\right)\ket{\phi^+_{A'E'}},\quad \beta=\bra{\xi_{A''E''}}\Gamma_{A''}^2\otimes\I\ket{\xi_{A''E''}}.
    \end{eqnarray}
   Here $\Gamma_{A''}=\bigotimes_k\mathcal{K}_{k,i_k}$ with $\mathcal{K}_{k,i_k}=\I$ for $i_k=0,1,2$ and $\mathcal{K}_{k,3}$ is the same operator as stated below Eq. \eqref{Yrob}. Using triangle inequality and then Eq. \eqref{Yrob} and also the fact that $|f^l_{i_1,\ldots,i_N}|\leq 1$, we obtain from \eqref{158}
    \begin{eqnarray}
\left|\bra{\phi^+_{A'E'}}\bra{\xi_{A''E''}}\tilde{\mathcal{A}}_{1,i_1}\otimes\bigotimes_{k=2}^N \mathcal{A}_{k,i_k}\otimes\Gamma_{A''}\otimes (R_{l|1})_E\ket{\psi_{AE}}-f^l_{i_1,\ldots,i_N}\right|\leq\sqrt{2}\varepsilon+|f^l_{i_1,\ldots,i_N}(1-\beta)| \leq (1+4N\gamma_N)\sqrt{2\varepsilon}.
    \end{eqnarray}
    For simplicity, we used $\varepsilon\leq\sqrt{\varepsilon}$.
    Now using triangle inequality, we obtain from \eqref{158} that
    \begin{eqnarray}
\left|\bra{\psi_{AE}}\tilde{\mathcal{A}}_{1,i_1}\otimes\bigotimes_{k=2}^N \mathcal{A}_{k,i_k}\otimes (R_{l|1})_E\ket{\psi_{AE}}-f^l_{i_1,\ldots,i_N}\right|\leq\Big|\ (\bra{\phi^+_{A'E'}}\bra{\xi_{A''E''}}-\bra{\psi_{AE}})\tilde{\mathcal{A}}_{1,i_1}\otimes\bigotimes_{k=2}^N \mathcal{A}_{k,i_k}\otimes (R_{l|1})_E\ket{\psi_{AE}}\ \Big|\nonumber\\+(1+4N\gamma_N)\sqrt{2\varepsilon}.\ \ \ 
    \end{eqnarray}
    Upper bounding the term on the right-hand side using Cauchy-Schwarz inequality along with the fact that $R_{l|1}^{\dagger}R_{l|1}\leq\I$ gives us
    \begin{eqnarray}\label{161}
    \left|\bra{\psi_{AE}}\tilde{\mathcal{A}}_{1,i_1}\otimes\bigotimes_{k=2}^N \mathcal{A}_{k,i_k}\otimes (R_{l|1})_E\ket{\psi_{AE}}-f^l_{i_1,\ldots,i_N}\right|\leq(1+4N\gamma_N)\sqrt{2\varepsilon}+\left\|\bra{\phi^+_{A'E'}}\bra{\xi_{A''E''}}-\bra{\psi_{AE}}\right\|\leq \sqrt{2\varepsilon}+f_1(\varepsilon).\nonumber\\
    \end{eqnarray}
    Again using triangle inequality in the left-hand-side of the above expression \eqref{161}, we express it as
    \begin{eqnarray}
        \left|\bra{\psi_{AE}}\tilde{A}_{1,i_1}\otimes\bigotimes_{k=2}^N \mathcal{A}_{k,i_k}\otimes (R_{l|1})_E\ket{\psi_{AE}}-f^l_{i_1,\ldots,i_N}\right|\leq \left|\bra{\psi_{AE}}(\tilde{\mathcal{A}}_{1,i_1}-\tilde{A}_{1,i_1})\otimes\bigotimes_{k=2}^N \mathcal{A}_{k,i_k}\otimes (R_{l|1})_E\ket{\psi_{AE}}\right|\nonumber\\+(1+4N\gamma_N)\sqrt{2\varepsilon}+f_1(\varepsilon).
    \end{eqnarray}
    Using Cauchy-Schwarz inequality to upper bound the first term on the right-hand side, we obtain 
    \begin{eqnarray}\label{195}
        \left|\bra{\psi_{AE}}\tilde{A}_{1,i_1}\otimes\bigotimes_{k=2}^N \mathcal{A}_{k,i_k}\otimes (R_{l|1})_E\ket{\psi_{AE}}-f^l_{i_1,\ldots,i_N}\right|\leq (1+4N\gamma_N)\sqrt{2\varepsilon}+f_1(\varepsilon)+\left\|(\tilde{\mathcal{A}}_{1,i_1}-\tilde{A}_{1,i_1})\ket{\psi_{AE}}\right\|%\nonumber\\ \leq (1+4N\gamma_N)\sqrt{2\varepsilon}+f_1(\varepsilon)+f_2(\varepsilon).
    \end{eqnarray}
    Let us now obtain an upper bound on $\left\|(\tilde{\mathcal{A}}_{1,i_1}-\tilde{A}_{1,i_1})\ket{\psi_{AE}}\right\|$. For this, we consider Eqs. \eqref{roburel11}, \eqref{roburel12} and \eqref{Yrob} and express them using the post-measurement state expression \eqref{54} as
    \begin{eqnarray}
        \Tr\left[(\tilde{\mathcal{A}}_{1,i_1}-\tilde{A}_{1,i_1}))R_{l|0}\proj{\psi_{AE}}\right]\leq\overline{P}(l|0) f_2(\varepsilon)
    \end{eqnarray}
    Summing over all $l$ in the above expression and then using triangle inequality, we obtain
    \begin{eqnarray}
       \left\|(\tilde{\mathcal{A}}_{1,i_1}-\tilde{A}_{1,i_1})\ket{\psi_{AE}}\right\|= \Tr\left[(\tilde{\mathcal{A}}_{1,i_1}-\tilde{A}_{1,i_1}))\proj{\psi_{AE}}\right]\leq f_2(\varepsilon).
    \end{eqnarray}
    Thus, from Eq. \eqref{195} we have that
    \begin{eqnarray}
        \left|\bra{\psi_{AE}}\tilde{A}_{1,i_1}\otimes\bigotimes_{k=2}^N \mathcal{A}_{k,i_k}\otimes (R_{l|1})_E\ket{\psi_{AE}}-f^l_{i_1,\ldots,i_N}\right| \leq (1+4N\gamma_N)\sqrt{2\varepsilon}+f_1(\varepsilon)+f_2(\varepsilon).
    \end{eqnarray}
    Continuing in a similar manner for every observable $A_{i,j}$, we obtain
    \begin{eqnarray}
        \left|\bra{\psi_{AE}}\tilde{A}_{1,i_1}\otimes\bigotimes_{k=2}^N A_{k,i_k}\otimes (R_{l|1})_E\ket{\psi_{AE}}-f^l_{i_1,\ldots,i_N}\right|\leq (1+4N\gamma_N)\sqrt{2\varepsilon}+f_1(\varepsilon)+(N-n_3)8\sqrt{2\varepsilon}+ n_3v_N\sqrt{2\varepsilon}
    \end{eqnarray}
    where $n_3$ denotes the number of indices $i_k=3$ and $f_1(\varepsilon)=\sqrt{\varepsilon/2}(17N(N^2-1))^2$. This completes the proof.
\end{proof}

\end{document}